\pdfoutput=1
\documentclass[]{article}
\usepackage{cite}
\usepackage{amssymb,amsmath,amsthm}
\usepackage[a4paper, scale = 0.8]{geometry}
\usepackage[colorlinks,linkcolor=blue,anchorcolor=blue,citecolor=red]{hyperref}
\usepackage[dvipdfm]{graphicx}
\usepackage{float}
\usepackage{subfigure}

\newtheorem{myTheorem}{Theorem}
\newtheorem{myDefinition}[myTheorem]{Definition}
\newtheorem{myLemma}[myTheorem]{Lemma}
\newtheorem{myRemark}[myTheorem]{Remark}
\newtheorem{myCorollary}[myTheorem]{Corollary}
\newtheorem{myExample}[myTheorem]{Example}
\newtheorem{myProblem}[myTheorem]{Problem}
\newtheorem{myProposition}[myTheorem]{Proposition}

\newcommand{\Ito}{It\^{o}}
\newcommand{\Gateaux}{{G\^{a}teaux}}
\newcommand{\Doleans}{Dol\'{e}ans}
\newcommand{\RadonNikodym}{Radon-Nikod\'{y}m }

\DeclareMathOperator{\Var}{Var}
\DeclareMathOperator{\Std}{Std}
\DeclareMathOperator{\Cov}{Cov}

\DeclareMathOperator{\myd}{d\!}

\DeclareMathOperator{\lev}{lev}

\newcommand{\myangle}[1]{\left\langle {#1}\right\rangle}

\title{Continuous-Time Risk Contribution of the Terminal Variance and its Related Risk  Budgeting Problem}
\author{ Guangyan JIA\footnote{Zhongtai Securities Institute for Financial Studies, Shandong University, Jinan, P.R. China. Email: jiagy@sdu.edu.cn},  Mengjin ZHAO\footnote{Zhongtai Securities Institute for Financial Studies, Shandong University, Jinan, P.R. China. Corresponding author. Email: zhaomj@mail.sdu.edu.cn}}
\date{\today}

\begin{document}
\maketitle

\begin{abstract}
    To achieve robustness of risk across different assets, risk parity investing rules, a particular state of risk contributions, have grown in popularity over the previous few decades.
    To generalize the concept of risk contribution from the single-period case to the continuous-time case,  we characterize the terminal variance's marginal risk contribution as a process through the {\Gateaux} differential and {\Doleans} measure. 
    Meanwhile, the risk contributions we extend here have the aggregation property. 
    Total risk can be represented as the aggregation of individual risk contributions among different assets and $(t,\omega)\in[0,T]\times \Omega$. 
    Subsequently, as an inverse target --- allocating risk, the risk budgeting problem of how to obtain policies whose risk contributions coincide with pre-given risk budgets in the continuous-time case is also explored in this paper. 
    These risk-budgeted policies are related to stochastic convex optimizations parametrized by associated risk budgets.
    We also give the economic interpretations on this optimization objective.
    On the application side, we put three examples to state the connections with other strategies regarding risk contribution/budgeting.
\end{abstract}

\section{Background}
\label{Section:Background}

Since the first ten years of this century, risk budgeting models have gotten much attention from academics and practitioners.
These portfolio construction methods do not require expected returns as inputs and calculate asset positions by appointing risk contributions under a predefined risk measure.

Just as Pearson stated in \cite{pearson_risk_2002}, risk budgeting acts as a 
\begin{quotation}
	`process of measuring and decomposing risk, using the measures in asset-allocation decisions, assigning portfolio managers risk budgets defined in terms of these measures, and using these risk budgets in monitoring the asset allocations and portfolio managers'.
\end{quotation}
Specifically, we need to explain the term risk contribution before considering the term budget.
This concept implies how we decompose the total risk into individual ones in some manners.
The contributions to risk can be classified into two types: risk contributions to the individual assets in the portfolio or specified risk factors.
In the second type, the returns of assets are often expressed in a linear model by factors such as market, momentum and etc.
However, this linear structure has significant limitations, and there is no definite conclusion as to which factors are recognized as effective.
We are studying the first type in this paper, which focuses only on the assets.
For a specific asset, the risk contribution is the quantity times the marginal growth of risk.
The quantity is the weights or shares on this asset.
The marginal growth of risk is calculated as the derivative of total risk with respect to the position of this asset.
When the predefined risk measure is positive homogeneous, the total risk coincides with the aggregation of individual risk contributions.
So we can use this tool, risk contribution,  to calculate how much the asset contributes to the total risk.
In this heuristic manner, people can construct a risk-budgeted portfolio in terms of risk, not weights, to avoid concentration on the risk.
Among those budget-based models, risk parity is one of the most common cases, equalizing individual risk contributions. 
We will give a formal description of related concepts such as risk contribution in Section~\ref{Section: Contribution}. The following example will be helpful to understand these concepts quickly.

\begin{myExample}
	\label{example:single-period}
	Suppose there are $d$ assets with their $\mathbb R^d$-valued random returns $r$ and the positive definite covariance matrix is denoted as $\Lambda \in \mathbb R^{d\times d}$.  We use $w\in \mathbb R^d$ to denote the weights on the assets, which describes how the total capital is allocated. To state the risk contribution of each asset, we set the standard error as the risk measure to quantify the total risk. The measure $f$ satisfies the Euler property
	\begin{equation}
	\label{fml:eular_in_example}
	f(w) := \sqrt{w^\top \Lambda w} = \sum_{i = 1}^d w_i \dfrac{\partial f}{\partial w_i}.
	\end{equation}
	Then the marginal risk contribution of the $i$-th assets is defined as $\dfrac{\partial f}{\partial w_i}$ and the risk contribution is defined as $w_i \dfrac{\partial f}{\partial w_i}$. Given a risk budget vector $\beta \in \mathbb R_+^d$ which represents the pre-determined relative risk level of the underlying assets, the $\beta$-budgeted portfolio should satisfy $(w_i \dfrac{\partial f}{\partial w_i})/(w_j \dfrac{\partial f}{\partial w_j}) = \beta_i /\beta_j$ for arbitrary $i,j$. As a special case, the risk contribution of risk parity portfolio is equally distributed, namely $w_i \dfrac{\partial f}{\partial w_i} = w_j \dfrac{\partial f}{\partial w_j}$ for arbitrary $i,j$. One can find the optimal solution $w^\star$ through some algorithms to minimize the heuristic cost function
	\begin{equation}
	\label{fml:loss_in_example}
	J(w) = \sum_{i,j}\left (w_i \dfrac{\partial f}{\partial w_i}/\beta_i - w_j \dfrac{\partial f}{\partial w_j}/\beta_j\right )^2
	\end{equation}
	or some other similar functions. 
	
	Here is a concrete calculation with the covariance matrix 
	$\Lambda = \begin{bmatrix}
		0.0900 & 0.0480 & 0.0225 \\ 
		0.0480 & 0.0400 & 0.0090 \\ 
		0.0225 & 0.0090 & 0.0225 \\ 
			\end{bmatrix}$. 
	Figure~\ref{fig:single-priod-example} gives the comparison on the weights and the risk contributions among equally weighted, risk parity, and minimized variance portfolios. 
	It shows that the minimized variance portfolio is aggressive with a negative weight on the first asset and implies the concentration in terms of risk contribution on the latter two assets. 
	However, because of the existence of correlation, the equally weighted strategy also shows the risk concentration on the first asset.
	Compared with the other two strategies, the risk parity strategy behaves more robust both on the weight term and the risk contribution term.
	\begin{figure}[H]
		\centering 
		\subfigure[Weights of Assets]{
			\includegraphics[width = 3.2 in]{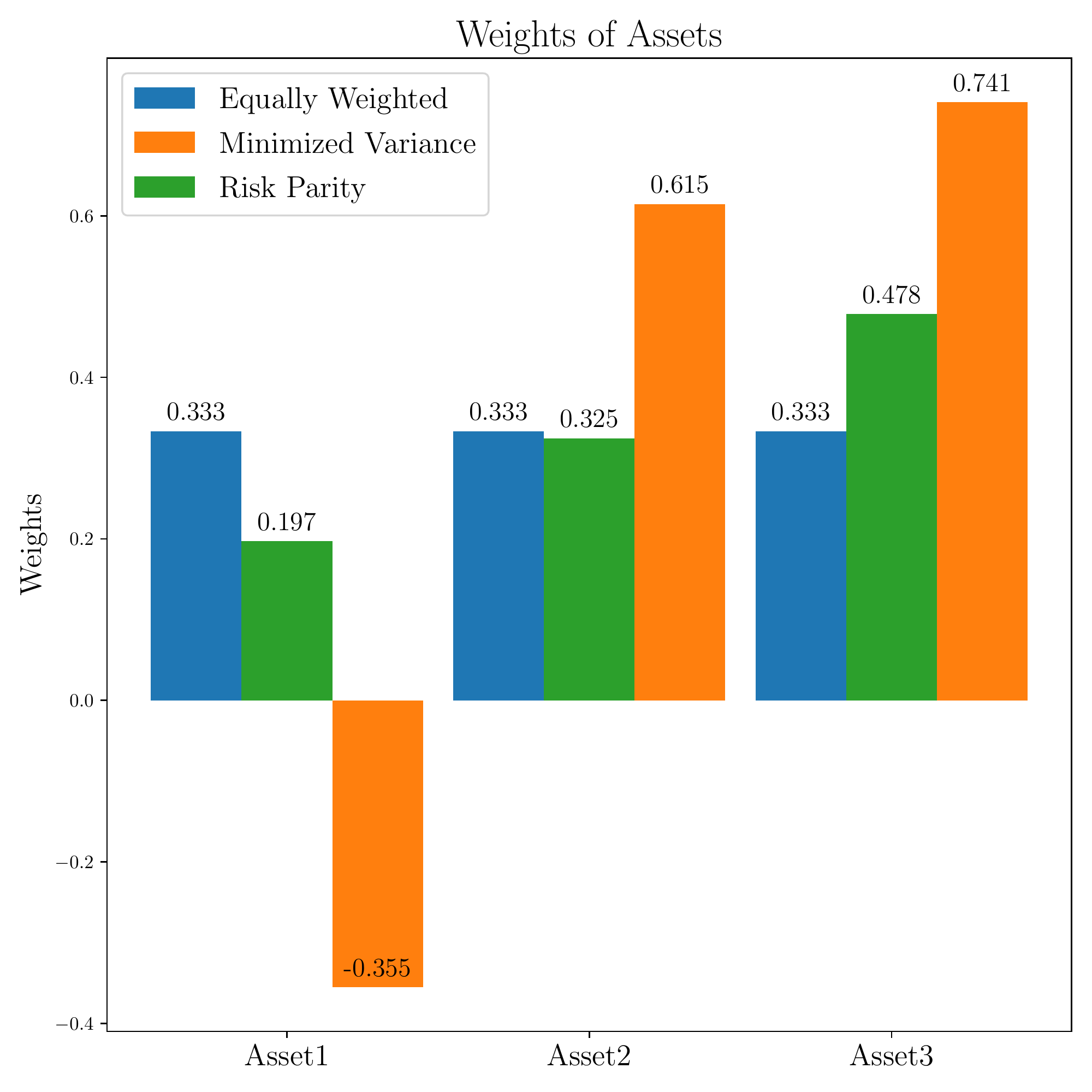}}
		\subfigure[Risk Contribution of Assets]{
			\includegraphics[width = 3.2 in]{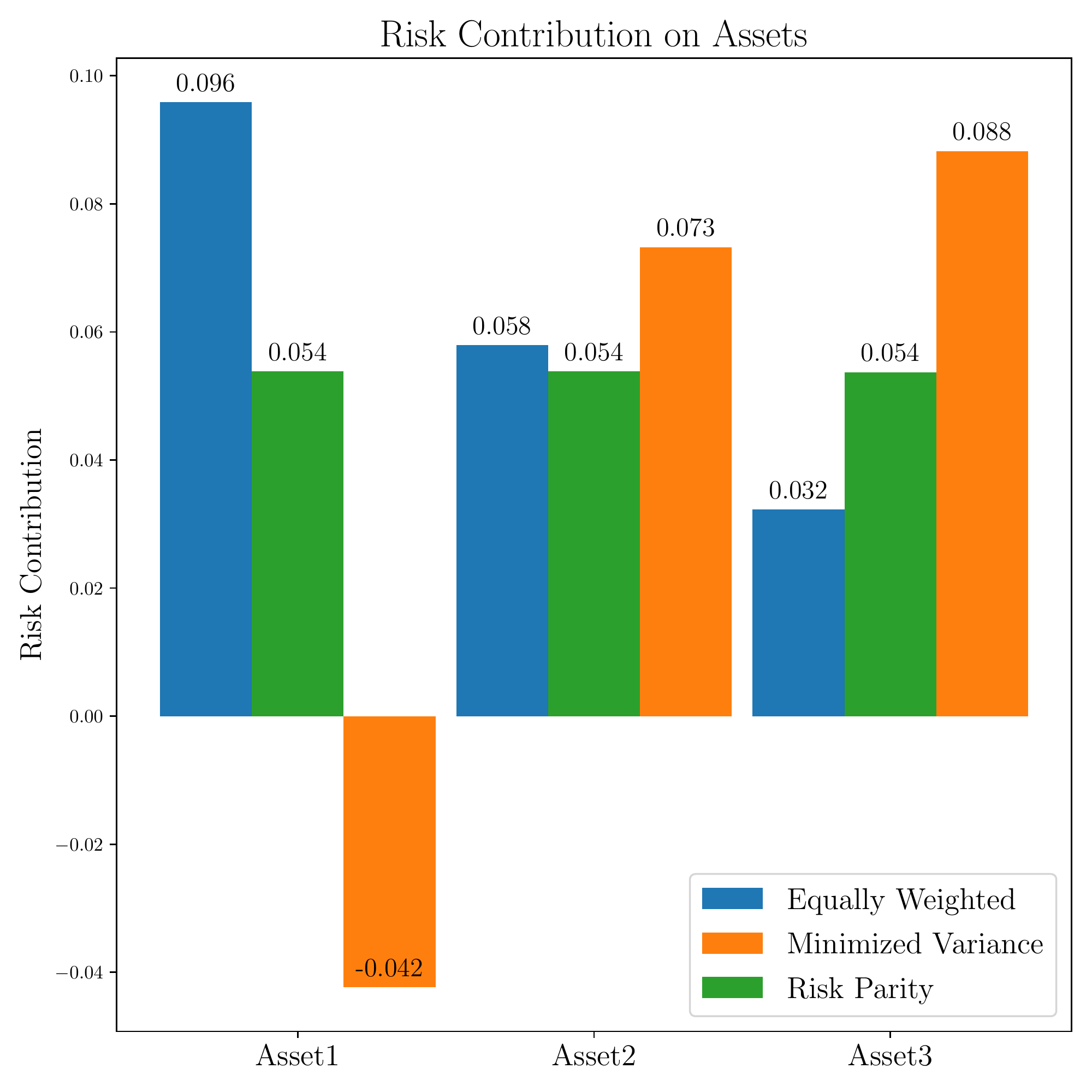}}
		\caption{Comparison of Three Strategies}
		\label{fig:single-priod-example}
		\end{figure} 
\end{myExample}

As a risk-based strategy, the idea of risk budgeting can date back to Markowitz's Nobel-prize-winning result, which quantifies the risk and gives the optimal solution of the mean-variance tradeoff framework.
Along with Markowitz's pioneering work, Merton's analytic answer to the efficient frontier \cite{merton1972analytic} is regarded as the essential foundation for portfolio design in a single period.
After this, kinds of literature on multi-period portfolio selections were dominated for years by the maximizing expected utility framework, and we recommend these works \cite{samuelson1975lifetime,merton1992Continuous}.
Similar to the expectation-utility framework, the continuous-time mean-variance problem was obtained by \cite{li_optimal_2000, zhou2000continuous} using the stochastic linear-quadratic framework with an embedding technique.
These works are more about maximizing utility or balancing return and risk.
Sensitivity analyses were performed \cite{best1991sensitivity1, best1991sensitivity2} on mean-variance portfolio selections that were substantially quadratic programming, meaning a loss of robustness when meeting non-stable inputs.
However, the minimization of the total risk always implies the concentration on weights.
Based on the market features and underlying assets, \cite{fernholz_stochastic_2002} shows that portfolios with diversification can beat the market portfolio in the stochastic language.
In light of policy constraints, one may seek a compromise between diversification and variance minimization objectives.
$1/N$-weighted, namely the equally weighted portfolios, can be considered an absolute diversification policy without knowledge on risk; and Markowitz's global variance minimized portfolio based on its covariance matrix is verified to take concentrations on those low-correlation assets.
Mixtures, such as portfolios with maximal diversification, equal risk contribution, and inverse-volatility weights, come between these two extremes.
The term risk parity was coined by \cite{qian2005risk}, and the analytical properties of risk parity portfolios were explored in \cite{maillard2010properties}.
They showed that the long-only risk parity portfolio is unique and lies between the minimized variance and equally weighted portfolios.
Risk parity and risk-budgeted portfolios are also widespread in practice.
Bridgewater Associates were the first to launch a risk parity fund on the asset management side. 
Equal risk contributions indices are also available (EURO iSTOXX 50 Equal Risk, for example).

The reasons for using risk budgeting strategies are diverse.
Paying more attention to the total risk, portfolio construction techniques based on decreasing risk often give results concentrating on low-risk assets.
Consequently, the mean-variance framework of Markowitz is not robust when the inputs are unstable.
Given this situation, people prefer to give regular constraints on the positions.
The second reason is the low-risk anomaly.
Empirical literature states that stocks with the lowest risks have higher returns than stocks with the sub-lowest risks.
This phenomenon occurs in the cross-section and in the time axis.
Risk parity strategies may be motivated through the low-risk anomaly.
The last reason can be attributed to the information capacity.
Forecasting the future returns is not easy, and no expectation of future returns is needed in the risk budgeting world.
Equal distribution on risk can give true diversification and thus should prevent the portfolio from losing better than other portfolios.

As the result of single-period investments, the concept risk contribution acts as a tool to detect if the portfolio is concentrated in risk.
Respectively, continuous-time investments should also have this tool to measure the concentration on their risk.
Unlike the single-period case, the concept of risk contribution may depend not only on assets but also on the stochastic intervals, i.e., the subsets of $\Omega \times [0,T]$.
For example, the single-period case can not answer this question: how much does the $i$-th asset contribute to the total risk when $S^i$ values in $[s_1,s_2]$ with $t\in [t_1,t_2]$?
Respectively, the risk budgets are supposed to be allocated over $\Omega \times [0,T]$.
In other words, the changes in the measurability of assets' dynamics, risk contributions, and risk budgets are essential to risk contribution/budgeting problems, and these changes do not appear in the single-period case.
We notice that the concentration occurs in the cross-section and on the time axis. 
We think continuous-time risk budgeted investments are attractive to catch the extra benefits from the low-risk anomaly phenomenon.
As far as we know, there is no such technique to define the risk contribution in the continuous-time case. It is interesting to generalize (marginal) risk contribution in the framework of stochastic analysis.

Inspired by the work of \cite{zhou2000continuous}, where the continuous-time mean-variance strategy is studied in a stochastic control framework, and that of \cite{maillard2010properties}, where the single-period risk budgeting problem can be solved by a variance optimization with some constraints within the covariance matrix, we want to develop the risk budgeting theory in the continuous-time case.
As a prerequisite, the concept of continuous-time (marginal) risk contribution should be adequately defined.
As the echo of how single-period risk contribution/budgeting works, this paper addresses the following points:
\begin{itemize}
	\item (Defining Risk Contribution) Without the structure of terminal wealth in the form of a covariance matrix, how can the risk contribution of policies be quantified? Additionally, does the risk aggregation property hold in the continuous-time case in which the policy is no longer a constant?
	\item (Posing Risk Budgeting) Assuming that risk contributions and budgets are stochastic processes, can we formulate and solve the inverse problem of how to align the desired policy with the pre-determined risk budget?
	\item (Connecting) What is the relationship between continuous-time and single-period solutions?
\end{itemize}
Also, we put three examples to show how the risk contribution/budget is linked to some known results in the continuous-time case.

The remainder of this paper is structured as follows.
In Section~\ref{Section: Contribution}, we will give some necessary axiomatic introduction to these concepts appearing in Example~\ref{example:single-period}.
Then we use the {\Gateaux} derivative of the terminal variance to derive a continuous process and define it as the marginal risk contribution.
Immediately afterward, the risk contribution was also appropriately defined.
In Section~\ref{Section: Budgeting}, we provide a convex stochastic optimization problem similar to the result \cite{maillard2010properties} presented in the single-period case, and the optimal policy to the problem happens to match the pre-given risk budget process.
In addition, we give some explanations for this optimization.
Moreover, classic single-period solutions are interpreted here as the solutions to those optimizations associated with the projection of the original risk budget process onto the naive $\sigma$-algebra.
We build a policy congruent with the work in \cite{moreira2017volatility} using regularised risk-parity optimization in Subsection~\ref{Section: JF} to demonstrate how the risk budgeting strategy behaves on the time axis.
Their volatility-timing strategy, motivated by the factor Sharpe ratio, reflects the robustness of the risk basket on the time axis and essentially is a risk parity strategy in our risk budgeting view.
To illustrate, we will cover risk contributions, risk-parity budgeting, and naive projections for the SABR model as an example in Subsection~\ref{Section: SABR}.
Coefficients in the SABR model setup can essentially influence the risk contribution of investment policies.
The influence of the coefficients states the connection between risk contribution and options of the underlying assets.
In Subsection~\ref{Section: Continuous-time MV Example}, we will give the risk contribution of continuous-time mean-variance strategies admitted by \cite{zhou2000continuous} and the result shows that there are concentrations in terms of risk contribution for these mean-variance strategies.
Finally, in Section~\ref{Section: Concluding}, we conclude with a discussion of the enhancements from the classic case and suggest several potential limitations as open problems.

\section{Characterize Risk Contribution} 
\label{Section: Contribution}
\subsection{A Formal Revisiting of Single-Period Risk Contribution/Budgeting}
This section extends the concept of risk contribution for the terminal variance to the continuous-time case, in which the risk contribution is characterized as a continuous process. 
In order to address fundamental questions regarding both the single-period case and the continuous-time case, it is imperative to have a prior and formal understanding of the tenets of risk contribution. 
To this goal, we address three concepts: risk measure, marginal risk contribution, and risk contribution.

Let $\rho(X)$ be the risk measure of the random variable $X$ on the probability space $(\Omega, \mathcal {F}, \mathbb {P})$, which is the position of the investment. 
Throughout this paper, we consider the variance, or standard deviation, as the risk measure.
In order to be acceptable in terms of a risk allocation principle, we give the axiomatic definition of the deviation risk measure, of which standard deviation is a particular case.
\begin{myDefinition}[Deviation Risk Measure,\cite{rockafellar_generalized_2006}]
	\label{def:DevMeasure}
	A deviation risk measure is a functional $\rho_D:L^2(\Omega;\mathbb {P}) \to [0,\infty]$ satisfying:
	\begin{enumerate}
		\item [D1] $\rho_D(h X) = h \rho_D(X)$ for all $h >0$;
		\item [D2] $\rho_D(X + Y)\leq \rho_D(X) + \rho_D(Y)$ for $X, Y \in L^2(\Omega;\mathbb {P})$;
		\item [D3] $\rho_D(X+C) = \rho_D(X)$ for $X\in L^2(\Omega;\mathbb {P})$ and $C\in \mathbb R$;
		\item [D4]$\rho_D(C) = 0$ for $C$ an arbitrary constant, and $\rho_D(X)>0$ for $X$ any non-constant random variable.
	\end{enumerate}
\end{myDefinition}

Suppose that there are $d$ assets in the market with the $\mathbb R^d$-valued random vector $r$ acting as the value of those assets and with the shares $\mathbb R^d$-valued $u$ investors hold. 
Then we have the random variable $X = u^\top r$ as the associated investment outcoming.
Hence we can talk about the sensitivity analysis of the risk measure.
This idea is based on marginal analysis and was proposed by \cite{litterman1996hot}.
Indeed, for a differentiable risk measure $\rho$, the marginal growth with respect to the $i$-th asset is calculated as $ \dfrac{\partial \rho(u^\top r)}{\partial u_i}$ and implies an infinitesimal property
\begin{equation}
\label{eq:infinitesimal-property}
\rho (w^\top r + \theta \delta_i) = \rho(w^\top r ) + \theta  \dfrac{\partial \rho(u^\top r)}{\partial u_i} + o(\theta).
\end{equation}
If the risk measure is positive homogeneous, then the function $u\mapsto \rho(u^\top r)$ is also positive homogeneous and implies the Euler homogeneity principle
\begin{equation}
\label{fml:EularDecomposition}
\rho(u^\top r) = \myangle{u,\dfrac{\partial \rho(u^\top r)}{\partial u}} = \sum_{i = 1}^d u_i\dfrac{\partial \rho(u^\top r)}{\partial u_i}.
\end{equation}
\cite{kalkbrener_axiomatic_2005} shows that this Euler allocation principle is the only risk allocation method in an axiom system for capital allocation.
Realizing these things,  we can show the definition of risk contribution.
\begin{myDefinition}[Risk Contribution/Budgeting - Single Period]
	\label{def:S-Contribution}
	Let the function $u\mapsto \rho(u^\top r)$ be a continuously differentiable risk measure(not necessary a deviation risk measure).
	\begin{itemize}
		\item Marginal risk contribution of portfolio $u$ in a vector style is defined by
		      $$
			      c^u:=\dfrac{\partial \rho (v^\top r)}{\partial v}\Big | _{v = u}
		      $$
		      where the $i$-th element is the marginal risk contribution on $i$-th asset.
		\item Risk contribution of portfolio $u$ in a vector style is\footnote{We denote $a^{(i)}$ as the $i$-th element of the vector $a$. With $\mathbb R^d$-valued $a$ and $b$, vector $a\odot b$ is defined by $a\odot b := [a^{(1)}b^{(1)},\dots, a^{(d)}b^{(d)}]^\top$.}
		      $$
			      k^u := u\odot c =  u\odot \dfrac{\partial \rho (v^\top r)}{\partial v}\Big | _{v = u}
		      $$
		      of which the $i$-th element is the risk contribution on the $i$-th asset.
	\end{itemize}
	Particularly,  a policy $u$ is called risk-parity if $k^u = u\odot c^u = \lambda e_d$ with notation $e_d = [1,...,1]^\top $ for some $\lambda \in (0,\infty )$. In other words, the risk contributions on these assets are equal
	\begin{equation}
		\label{fml:S-ERCP}
		k ^{(i)} = u^{(i)}c^{(i)} = u^{(j)}c^{(j)} = k^{(j)} =  \lambda , \text{ for each } i,j.
	\end{equation}

	When exogenously given a budget $\beta \in \mathbb {R}_+^d$, the risk budgeting problem aims to find a suitable policy $u^\star$ satisfying
	\[u^\star \in \Big\{u\Big | u\odot \dfrac{\partial \rho (v^\top r)}{\partial v}\Big | _{v = u} = \beta  \Big \}.\]
\end{myDefinition}

Obviously, from the above point of view, the variance does not meet those properties deviation risk measures. 
However, Euler's aggregation principle, together with the marginal contribution, still holds and is similar to standard deviations.
Standard deviation is a deviation risk measure, and we can see 
\[u\odot \dfrac{\partial \Var(u^\top r)}{\partial u} = 2\Std (X) \left[ u\odot \dfrac{\partial \Std (u^\top r)}{\partial u} \right]\]
and 
\[\myangle{u, \dfrac{\partial \Var(u^\top r)}{\partial u}} = 2\Var (X).\]
The risk contribution of variance is the scaled one of the standard deviation measure with the coefficient $\dfrac{1}{2\Std (X)}$ and the proportion of $i$-th varaince contribution to the standard deviation contribution is identical to that of $j$-th asset
\[
\left[u\odot \dfrac{\partial \Var(u^\top r)}{\partial u}\right]_i\Big /\left[u\odot \dfrac{\partial \Var(u^\top r)}{\partial u}\right]_j = \left[ u\odot \dfrac{\partial \Std (u^\top r)}{\partial u} \right]_i\Big /\left[ u\odot \dfrac{\partial \Std (u^\top r)}{\partial u} \right]_j.
\]
Euler's principle is also valid for the variance with the scaled coefficient $\dfrac{1}{2}$.
In this sense, it is equivalent to considering the risk contribution of the variance. 
For convenience, we can set the marginal risk contribution of the variance as $\dfrac{1}{2}\dfrac{\partial \Var(u^\top r)}{\partial u}$  and will talk about the continuous-time risk contribution of the terminal variance in the following sections.

\subsection{Technical Preliminaries}
We will characterize the associated risk contribution and marginal risk contribution of the terminal variance acting as the risk measure. 
The marginal risk contribution of a given policy will be firstly derived in a signed measure form whom the integral of the policy with respect to is the terminal variance. 
Then the \RadonNikodym derivative of this signed measure is proved to be a continuous process meaning that total risk is continuously aggregated. 
Some preliminary settings are listed below.

With the time horizon $[0,T]$, we let $(\Omega, \mathcal F, \mathbb P)$ be the random basis equipped with a filtration $\mathbb F = \{\mathcal F_t\}_{t\in [0,T]}$ and $\{B_t\}_{t\in [0,T]}$ is an $m$-dim Brownian motion on this space. 
We also assume that $\mathbb F$ is the natural filtration generated by the Brownian motion $B$ and completed by the collection of $\mathbb P$-null sets $\mathcal N$, i.e., 
$$
\mathcal F^0_t = \sigma (B_s,s\leq t); \mathcal F_t = \mathcal F_t^0\vee \mathcal N
$$
with $\mathcal F = \mathcal F_T$.
The processes we discussed in this paper are assumed to be $\mathbb F$-adapted.  
We also denote $\Sigma_p$ the predictable $\sigma$-algebra  and $H\circ X$ the stochastic integral. 
We use the notation $L^p = L^p(\Omega,\mathcal F;\mathbb{ P })$ for the space of random variables.

\paragraph{Assets}
The value processes of assets are characterized as a sequence of continuous special semi-martingales in 
\[\mathcal H^2_{\mathcal S} := 
\left \{X\left|   \Big \lVert \left\langle M\right\rangle^{1/2}_T\Big \rVert_{L^2} + \Big \lVert\int_0^T |\myd A_s|\Big \rVert_{L^2}< \infty , \text{ with the carnonical decomposition }X = X_0 + A + M.\right.\right \}\]
to make sure that the value of assets admits a finite variance.
Without the loss of generality, the dynamics of assets value are assumed to be in the following form
\[
\label{fml:dynamic_of_assets}
\left \{
\begin{aligned}
\myd & S^{(i)}_t = b^{(i)}_t\myd t + {\sigma_t^{(i)}} \myd B_t\\
 & S^{(i)}_0 = s_0^{(i)}
\end{aligned}
\right., \mbox{ for } i = 1,\dots,d
\]
where the instantaneous diffusion is a multi-dim process $\sigma^{(i)}_t$, the $i$-th row of $\sigma _t = [\sigma^{(i,j)}_t]_{i,j = 1}^{d,m}$. We shall denote the $\mathbb R ^d$-valued process $S = [S^{(1)}, S^{(2)},\dots, S^{(d)} ]^\top$ as the assets for short.

\paragraph{Policy/Control}
The policy $u$ is assumed to be a $\mathbb R^d$-valued predictable process with $u\in \mathbb S^\infty$ where
\[\mathbb S^\infty = \Big\{X\Big|\big \lVert X\big \rVert_{\mathbb S^\infty}:=\big \lVert\sup_{s\leq T}|X_s|\big \rVert_{L^\infty}<\infty\Big\}.\] 
For every $t$, the random variable $u_t$ is the shares we buy at time $t$, and process $u \in \mathbb S^\infty$ implies that we cannot afford to hold infinitely many shares in any case.
\paragraph{Investment Process}
The {\Ito}-type integral of the policy $u$ with respect to the assets $S$ is defined as the value of investment $X^u$ where 
\begin{equation}
\label{eq:investment}
X^u_t =x_0+  (u ^ \top \circ S)_t =x_0 + \int_{0}^{t}u_\tau ^\top\myd S_\tau\text{ , for } t \in [0,T]
\end{equation}
is the value of investment $X^u$ at time $t$ with its initial wealth $x_0$. By the way, the investment is called a self-financing portfolio in addition if
\begin{equation}\label{eq:portfolio}
X_t^u = u_t^\top S_t = \sum_{i = 1}^d u^{(i)}_tS^{(i)}_t
\end{equation}
for each $t$. 
With $u\in \mathbb S^\infty$ and $S \in \mathcal H ^2_{\mathcal S}$, it's easy to check $X^u$ is in $\mathcal H ^2_{\mathcal S}$, and we can say that the investment $X^u$ is still variance-finite. 

\paragraph{Terminal Variance as the Risk Measure}
Here we consider the terminal variance $\text{Var}((u\circ S)_T)$ of the portfolio $X^u$ as its risk measure where
$$
\text{Var}((u\circ S)_T) = \mathbb E[(u\circ S)_T^2]- (\mathbb E (u\circ S)_T)^2.
$$
It's easy to check that  $u \mapsto \Var((u\circ S)_T)$ is a convex functional. Moreover, the space $\mathcal H^2_{\mathcal S}$ ensures that the investment process $X^u$ has a second-order moment implying a finite terminal variance.
\begin{myDefinition}[Non-Degeneration]
	\label{def:ND-0}
	The market is non-degenerate if we always have 
	$$
	\text{Var}((u\circ S)_T) >0
	$$
	for arbitrary policy $u\neq 0$.
\end{myDefinition}
With the slope mapping $\theta \mapsto \dfrac{1}{\theta}[\text{Var}(X_T^{u+\theta v}) - \text{Var}(X_T^{u})] = 2\text{Cov}(X_T^u, X_T^v)+ \theta \text{Var}(X_T^v)$ at the point $u$ in the direction $v$, this non-degeneration condition above ensures that the functional $u\mapsto \text{Var}(X_T^u)$ is  strictly convex.

\subsection{Marginal Risk Measure}
At the first of this subsection, we put $d = 1$, i.e., there is only one risky asset in our sight. 
The result of the $1$-dim case can help us understand how risk contribution behaves on the time axis. 
As an application, volatility-managed portfolios in Subsection~\ref{Section: JF} can be derived in this risk contribution view. 
With the help of the differential technique in this subsection, the multi-dimensional result will not be lengthy when it meets the cross-section risk contributions of various assets.
 
By {\Ito}'s formula, we can rewrite the terminal variance of our investment process $u\circ S$ and divide it into three parts,
\begin{align*}
&\text{Var}\left ((u\circ S)_T\right) \\
=& \mathbb E\left [(u\circ S)_T^2- (u\circ S)_T \int_\Omega (u\circ S)_T(\omega)\mathbb P(\myd \omega)\right]\\
=& \mathbb E\Bigg[\underbrace{\int_0^T 2(u\circ S)_t u_t\myd S_t}_{\text{I}} + \underbrace{\int_0^Tu^2_t\myd \langle  S\rangle_t}_{\text{II}} - \underbrace{\int_0^Tu_t\myd S_t\int_\Omega (u\circ S)_T(\omega')\mathbb P(\myd \omega')}_{\text{III}}\Bigg].
\end{align*}
For the (I) part, we give an operator 
$$
\Phi_{\text{I}} : u \mapsto  \int_0^T 2(u\circ S)_t u_t\myd S_t
$$
whose image is a random variable.
We can also define the operator $\Phi _2\text{ and } \Phi_3$ for the (II) and (III) parts respectively:
$$
\begin{aligned}
\Phi & _{\text{II}}: u \mapsto \int_0^Tu^2_t\myd \langle  S\rangle_t
\\
\Phi &_{\text{III}}: u \mapsto \int_0^Tu_t\myd S_t\int_\Omega (u\circ S)_T(\omega')\mathbb P(\myd \omega')
\end{aligned}
$$
The images of these three operators are equipped with the $L^1$-norm.

\begin{myLemma}[{\Gateaux} differential]\label{lm:gateaux_differential}
Under $L^1$-norm, the G\^{a}teaux differentials of $\Phi_{\text{I}}, \Phi_{\text{II}}$ and $\Phi_{\text{III}}$ are
\begin{align*}
\myd \Phi_{\text{I}}(u;v) = &\int_0^T 2(u\circ S)_tv_t + 2(v\circ S)_tu_t \myd S_t\\
\myd \Phi_{\text{II}}(u;v) = &\int_0^T 2u_tv_t \myd \langle S \rangle_t\\
\myd \Phi_{\text{III}}(u;v) = &(v\circ S)_T\int_\Omega (u\circ S)_T(\omega')\mathbb P(\myd \omega') + (u\circ S)_T \int_\Omega (v\circ S)_T(\omega')\mathbb P(\myd \omega').
\end{align*}
\end{myLemma}
\begin{proof}[Proof]
	With arbitrary $\theta \in [0,1]$ and $v \in \mathbb S^\infty$, we take the variation of $\Phi _x( x = \text{I,II,III})$.
For (I) part, we have
\begin{align*}
&\Phi_{\text{I}} (u+\theta v) - \Phi _{\text{I}} u \\
=& \int_0^T\Big [ 2\big ((u+\theta v)\circ S\big )_t(u+\theta v)_t - 2(u\circ S)_tu_t \Big ]\myd S_t\\
=& \int_0^T 2\Big [(\theta v\circ S)_tu_t + 2 (u \circ S)_t\theta v_t + 2(\theta v\circ S)_t\theta v_t \Big ]\myd S_t\\
=& \int_0^T \Big [2\theta \big [(v\circ S)_tu_t + (u\circ S)_tv_t\big ] + 2\theta ^2 (v\circ S)_tv_t \Big ]\myd S_t	
\end{align*}
by the linearity of stochastic integral.
Let $\mathcal L_{\text{I}} v := \int_0^T 2 [(v\circ S)_tu_t + (u\circ S)_tv_t] \myd S_t$. Obviously, $\forall v \in \mathbb S^\infty$ the linear operator $\mathcal L_{\text{I}} $ satisfies 
\begin{align*}
&\lim_{\theta \to 0} \left \lVert\dfrac{\Phi_{\text{I}} (u+\theta v) - \Phi_\text{I} u}{\theta} - \mathcal L_\text{I} v\right \rVert _{L^1} \\
= & \lim _{\theta \to 0}\left \lVert \int_0^T 2\theta (v \circ S)_tv_t\myd S_t\right \rVert_{L^1} \\
\leq & \lim _{\theta \to 0}\left \lVert (v\circ S)\right \rVert^2_{\mathcal H^2_{\mathcal S}} =  0 .
\end{align*}
Hence the linear operator $\mathcal L_\text{I}$ is the {\Gateaux} differential of $\Phi _\text{I}$.
Considering that $\mathcal L_\text{I} v$ is associated with $u$, we denote it by $\myd \Phi_\text{I}(u;v)$ in the  fashion of directional derivative.

Similarly, for the (II)part we have the convex variation of $\Phi _\text{II}$
$$
\Phi_\text{II}(u+\theta v) - \Phi _\text{II} u
 = \int _0^T 2\theta u_tv_t + \theta ^2 v_t^2 \myd \langle S \rangle_t
$$
and the {\Gateaux} differential of $\Phi_\text{II}$
$$
\myd \Phi _\text{II} (u,v) = \int_0^T 2u_tv_t \myd \langle S \rangle_t.
$$

As for part (III), we can get
\begin{align*}
&\Phi_\text{III} (u+\theta v) - \Phi _\text{III} u \\
=& (\theta v\circ S)_T \int_\Omega (u\circ S)_T(\omega') \mathbb{P} (\myd w')
+ (u\circ S)_T \int_\Omega (\theta v\circ S)_T(\omega') \mathbb{P} (\myd w')\\
&+ ((\theta v\circ S)_T \int_\Omega (\theta v\circ S)_T(\omega') \mathbb{P} (\myd w').	
\end{align*}
And we take 
$$
\mathcal{L} _\text{III} v := ( v\circ S)_T \int_\Omega (u\circ S)_T(\omega') \mathbb{P} (\myd w')
+ (u\circ S)_T \int_\Omega ( v\circ S)_T(\omega') \mathbb{P} (\myd w').
$$
Then we get 
$$
\lim_{\theta \to 0} \left \lVert\dfrac{\Phi_3 (u+\theta v) - \Phi_3 u}{\theta} - \mathcal A_3 v\right \rVert _{L^1} = 0
$$
since 
\begin{align*}
&\theta ^2 \left \lVert (v\circ S)_T\int_\Omega (v\circ S)_T(\omega') \mathbb{P} (\myd w')   \right \rVert_{L^1}
\\
\leq & \theta ^2 \mathbb E[|(v\circ S)_T|]\left | \int_\Omega (v\circ S)_T(\omega') \mathbb{P} (\myd w') \right |
\\
\leq & \theta ^2 \big ( \mathbb{E} [|(v\circ S)_T|] \big)^2 <\infty.
\end{align*}
Finally the {\Gateaux} differential of $\Phi_\text{III} $ at $u$ is $\myd \Phi_\text{III}(u;v) = \mathcal L_\text{III} v$.
\end{proof}

\begin{myDefinition}[{\Doleans} measure, \cite{cohen2015stochastic}]
	\label{defi:Doleans_measure}
	Given a non-decreasing integrable process $A$, there is a non-negative measure $\mu _A$ defined on $([0,T]\times \Omega, \mathcal B_{[0,T]} \otimes \mathcal F)$ as
	$$
	\mu_A(E) = \mathbb E [(1_E \circ A)_T] = \mathbb E\left[ \int_0^T 1_E\myd A_s \right]
	$$
	for each set $E\in \mathcal B\otimes \mathcal F$. The measure $\mu_A$ is the {\Doleans} measure associated with the process $A$.
\end{myDefinition}

We will propose a signed measure to ensure that the terminal risk is distributed appropriately on $\Omega \times [0,T]$.

\begin{myTheorem}[Marginal Risk Measure - $1$-dim]
\label{thm:marginal_risk_measure_1d}
Given $d = 1$, the terminal variance of the investment $u\circ S$ can be represented as an integral
\begin{equation}
\label{fml:risk_gradient_measure}
\text{Var}((u\circ S)_T)  = \dfrac 1 2\int_{[0,T]\times \Omega}u(t,\omega)\mu(\myd t,\myd \omega)
\end{equation}
where $\mu(E) := \mathbb E [\myd \Phi_\text{I}(u;1_E) + \myd \Phi_\text{II}(u;1_E) + \myd \Phi_\text{III}(u;1_E)] \text{ for the sets } E\text{ in the predictable $\sigma$-algebra } \Sigma_p$. 
$\mu(E)$ is called the marginal risk measure\footnote{Measure $\mu$ is actually induced by $u$, therefore we use the notation $\mu^u$. 
Without causing ambiguity, we prefer $\mu$ omitting $u$.} of $\text{Var}((u\circ S)_T)$ at $u$. 

Moreover, in the non-degenerate market, the signed measure $\mu$ introduced by the mapping $u \mapsto \mu ^u$ above is unique on the support  $\mathrm{supp}(u) = \{u\neq 0\}$.
\end{myTheorem}
\begin{proof}[Proof]
	When we take the direction policy $v$ in the form of $1_E$ where $E$ is a $\Sigma _p$-measurable set, the {\Gateaux} differentials in the Lemma~\ref{lm:gateaux_differential}  give us three set functions with the common domain $\Sigma _p$
	\begin{align*}
	\mu _\text{I} (E):= & \mathbb E[\myd \Phi_\text{I}(u,1_E)],
	\\
	\mu _\text{II} (E):= & \mathbb E[\myd \Phi_\text{II}(u,1_E)],
	\\
	\mu _\text{III} (E):= & \mathbb E[\myd \Phi_\text{III}(u,1_E)]	.
	\end{align*}
	We want to show that $\mu _x(x = \text{I,II,III})$ are signed measures so that we can get the result on representation by integrating our policy $u$ with respect to the summation of these measures $\mu _x$. For the measures $\mu _x(x = \text{I,II,III})$, it is easy to check the following properties:
	\begin{itemize}
		\item $\mu_x (\varnothing) = 0$;
		\item $\mu_x ([0,T]\times \Omega) < \infty$;
		\item (Finite Additivity). For arbitrary disjoint $E_1, E_2 \in \Sigma_p$, we have
		$$\mu_x(E_1 + E_2) = \mu_x(E_1) + \mu_x(E_2).$$
	\end{itemize}
	And it is left us to seek the countably additivity property of $\mu_x$.
	
	For $\mu_\text{I}$ part, we should notice that
	$$
	\mu_\text{I} (E) = \mathbb E[\myd \Phi_\text{I}(u,1_E)] = \mathbb{E}\int_0^T \Big (2(u\circ S)_t1_E + 2(1_E \circ S)_tu_t\Big)\myd F_t
	$$
	where $F$ is the part of finite variation in the canonical decomposition of $S$. With the positive-negative decomposition of $u = u^+-u^-, F = F^+ - F^- $ and $  S = (F^+ + M^+) - (F^- + M^-) = S^+ - S^-$, we can rewrite the first term
	\begin{align*}
	&\nu_1(E) =  \mathbb E \int_{0}^T 2(u\circ S)_t 1_E\myd F_t
	\\
	=& \mathbb{E} \int_0^T 2 \big((u^+-u^-)\circ (S^+-S^-)\big)_t1_E\myd (F^+_t - F^-_t)
	\\
	=& \mathbb{E} \int_0^T  2\big( (u\circ S)_t^+ - (u\circ S)_t^-\big)1_E\myd (F^+_t - F^-_t)
	\\
	=& 2\mathbb E \int_0^T 1_E\Big((u\circ S)_t^+\myd F_t^+ + (u\circ S)_t^- \myd F_t^-\Big) - 1_E\Big((u\circ S)_t^-\myd F_t^+ + (u\circ S)_t^+ \myd F_t^-\Big).
	\end{align*} 
	By Definition~\ref{defi:Doleans_measure}, $\nu_1$ is a combination of {\Doleans} measures induced by four increasing process, hence a signed measure.
	And the decomposition of the second term(denoted by $\nu_2$) shows
	\begin{align*}
	&\nu_2 (E) = \mathbb E \int_0^T 2(1_E\circ S)_tu_t\myd F_t
	\\
	 = & 2\mathbb E \int_0^T \Big( 1_E\circ (S^+ - S^-) \Big)_t(u^+_t - u^-_t)\myd (F^+_t - F^-_t)
	 \\
	  = & 2\mathbb E \int_0^T \Big((1_E\circ S^+)_tu^+_tdF^+_t + (1_E\circ S^+)_tu^-_tdF^-_t + (1_E\circ S^-)_tu^+_tdF^-_t + (1_E\circ S^-)_tu^-_tdF^+_t \Big)
	  \\
	  &+ \Big((1_E\circ S^+)_tu^+_tdF^-_t + (1_E\circ S^+)_tu^-_tdF^+_t - (1_E\circ S^-)_tu^+_tdF^+_t + (1_E\circ S^-)_tu^-_tdF^-_t \Big).
	\end{align*}
	Noticing that $\big\{\sum_{i = 1}^n 1_{E_i} \big\}_n$ is a non-decreasing sequence for disjoint set sequence $\{E_i\}_i$, we take $f_n^{+,+} = \Big(\sum_{i = 1}^{n} 1_{E_i}\circ S^+\Big)u^+$, then $\Big\{ f_n^{+,+}\Big\}_n$ is also a non-decreasing sequence and respectively non-decreasing $\Big\{ -f_n^{+,-}\Big\}_n, \Big\{ -f_n^{-,+}\Big\}_n, \Big\{ f_n^{-,-}\Big\}_n$. 
	Applying monotone convergence theorem, we can get the countably additivity of $\nu_2$. 
	Consequently $\mu_1$ is a signed measure on $\Sigma _p$.
	
	As for $\mu_\text{II}$, it can be considered the difference of two {\Doleans} measures induced by two respected processes $u^+\circ \langle S\rangle$ and $u^-\circ \langle S\rangle$ since
	$$
	\mu_\text{II}(E) = 2\mathbb E\int_0^T 1_E u_t\myd \langle S\rangle _t = 2\mathbb E\int_0^T 1_E (u^+_t - u^-_t)\myd \langle S\rangle _t = 2\mu_{u^+\circ \langle S\rangle}(E) - 2\mu_{u^-\circ \langle S\rangle}(E).
	$$
	Hence $\mu_\text{II}$ is a signed measure.
	Through the decomposition technique we treat $\nu_2$ with, we can see that 
	$$
	\mu_\text{III} (E) = 2\mathbb E [(1_E\circ S)_T]\mathbb{E}[(u\circ S)_T] = 2\mathbb E \left [\big(1_E\circ (F^+ - F^-)\big)_T\right ]\mathbb{E}[(u\circ S)_T].
	$$
	Then $\mu_\text{III}$ is a signed measure on $\Sigma _p$.
	Finally we take $\mu:= \mu_\text{I} + \mu_\text{II} + \mu_\text{III}$, and the integral of $u$ with respect to $\mu$ shows us
	$$
	\int_{[0,T]\times \Omega} u(t,\omega) \mu(\myd t, \myd \omega) = 2\text{Var}(X_T^u).
	$$
	
	If there is an another signed measure $\mu'$ also representing the integral above, then we straightly have a singular property
	$$
	\int_{[0,T]\times \Omega} u(t,\omega)\myd (\mu - \mu') = 0
	$$
	for arbitrary $u$. 
	We can take $u$ in the form $u = \sum_{i}\lambda_i1_{E_i}$ for with $E_i\in \Sigma_p$, then the singular proprty implies $\mu(E) - \mu'(E) = 0$ for any arbitrary predictable set $E$.
	Hence $\mu$ and $\mu'$ must agree on $\Sigma_p$ unless $\Var (X^u) = \int u \myd \mu^u = \int u \myd \mu ' = 0$ with some special policies $u$.
	This singular result together with the non-degenerate condition in Definition~\ref{def:ND-0} ensures the uniqueness of mapping $u \mapsto \mu^u$.
\end{proof}
\begin{myRemark}
\label{rmk:marginal_as_dual}
The concept of marginal risk measure we put here has several meanings.
Given any $C^1$-continuous positive homogeneous function $f(x)$ of order $k$, we have the Euler's homogeneity theorem
$$
kf(x_0) = \left\langle\dfrac{\partial f}{\partial x}\Big |_{x = x_0},x_0\right\rangle
$$
where $\dfrac{\partial f}{\partial x}\Big |_{x = x_0} = \nabla f(x_0)$ is the gradient of $f$ at $x_0$.
The same is true for the marginal risk contribution in Definition~\ref{def:S-Contribution}.
Compared with the property above, Theorem~\ref{thm:marginal_risk_measure_1d} inherits that in the case $k = 2$. 
Further, this marginal risk measure is given by the {\Gateaux} differential which characterize the marginal growth of the terminal variance.
That is the reason why we stress the word marginal.

We should also notice that for arbitrary $v$ the linear mapping
$$
v \mapsto \mathrm{Cov}((u \circ S)_T, (v \circ S)_T)
$$
leads to a Riesz representation in the dual space of $v$, and hence the marginal risk measure $\mu^u$ is identical to this Riesz representation.
\end{myRemark}

\subsection{The Flow Representation of Marginal Risk Measure}
The risk associated with an investment process, specifically the terminal variance in our model, is supposed to be accumulated over the time interval $[0,T]$ and sample space $\Omega$ where the investment process behaves differently.
The terminal variance of our investment process should take the following suspected form
$$
\text{Var}(X^u_T) \overset{?}{=} \mathbb E \int_0^T u_t c_t\myd t
$$
where the measurable process $c(t,\omega)$ can be considered the instantaneous marginal risk contribution to the total risk.
The terminal variance has already been represented as the integral in Theorem~\ref{thm:marginal_risk_measure_1d}.
What we are going to do is the show the relation between $c$ and $\mu$.

\begin{myTheorem}[Flow Representation - $1$-dim]
	\label{thm:representation_of_RGM}
	Given a non-degenerate market, the signed measure $\mu$ on $\Sigma_p$ obtained in Theorem~\ref{thm:marginal_risk_measure_1d} can be uniquely represented as 
	\begin{equation}
	\label{fml:representation_of_RGM}
	\mu (E) = \mathbb{E} \int_0^T 1_E (t,\omega) \myd C(t,\omega) = \mathbb{E} \int_0^T 1_E(t,\omega) c(t,\omega)\myd t
	\end{equation}
	where the measurable set $E$ is taken in $\Sigma_p$ and the process $C$ is continuous in $\mathcal A^1$. As the instantaneous marginal risk contribution, the process $c$ is the \RadonNikodym derivative of $C$ with respect to $t$.
\end{myTheorem}
\begin{proof}
	The predictable $\sigma$-algebra $\Sigma_p$ is generated by the simple left-continuous processes  in the form $\phi 1_{(s,t]}$ where $s,t \in [0,T]$ and $\phi \in \mathcal{F}_s$. 
	We can just take the simple form $E = (s,t]\times H, H\in \mathcal F_s$ into consideration. 
	The measure of $E$ in the form above can be split as
	$$
	\mu((s,t]\times H) = \mu((s,T]\times H) - \mu((t,T]\times H).
	$$
	We can take a sequence of set functions parametrized by $t$ as
	$$
	m_t(H):= \mu((t,T]\times H)
	$$
	where $m_t$ is defined on the $\sigma$-algebra $\mathcal{F}_t$ for each $t$. 
	
	We claim that $m_t$ is a signed measure absolutely continuous with respect to $\mathbb{P} \big | _{\mathcal{F}_t}$ since it is induced by the signed $\mu$. 
	It's left to show the absolute continuity of $m_t$ for every fixed index $t$.
	For the $\mathbb{P}$-null sets $F\in \mathcal F_t$, we take $E_0 = (t,T]\times F$.
	Those three parts of $\mu$ in Theorem~\ref{thm:marginal_risk_measure_1d} satisfies the following properties.
	\begin{itemize}
		\item For $\mu_\text{I}$ part, we have
		\begin{align*}
		\mu_\text{I}(E_0) =& \mathbb{E} \int_0^T \Big(2(u\circ S)_\tau 1_{E_0}(\tau,\omega)+2(1_{E_0}\circ S)_\tau u_{\tau}\myd S_\tau \Big) 
		\\
		=& \mathbb{E} \int_0^T \Big(2(u\circ S)_\tau 1_{(t,T]\times F}(\tau,\omega)+2(1_{(t,T]\times F}\circ S)_\tau u_{\tau}\myd S_\tau \Big) 
		\\
		=& \mathbb{E} \Big(\int_0^T 2(u\circ S)_\tau 1_{(t,T]}(\tau)1_F(\omega)\myd S_\tau \Big) + 2\mathbb{E} \Big(\int_0^T   \big (1_F\circ (1_{(t,T]}\circ S)\big )_{\tau}u_{\tau}   \myd S_\tau \Big) 
		\\
		 = & \mathbb{E} \Big(1_F \int_t^T 2(u\circ S)_\tau \myd S_\tau \Big)
		  + \mathbb{E} \Big(1_F \int_t^T 2(1_{(t,T]}\circ S)_\tau u_{\tau} \myd S_\tau \Big)
		  \\
		   = & 0 
		\end{align*}
		since $F$ is $\mathcal F_t$-measurable.
		\item In a same manner, $\mu_\text{II}$ part gives
		$$
		\mu _\text{II} (E_0) = \mathbb E\int_0^T1_{E_0}(\tau,\omega)u_{\tau}\myd \langle S\rangle_\tau
		 = \mathbb{E} \Big(1_F\int_t^T u_{\tau}\myd \langle S\rangle_\tau \Big) = 0.
		$$
		\item  As for $\mu _\text{III}$ part, noticing $(1_{E_0} \circ S)_T = 1_F(S_T - S_t)$, we have 
		$$
		\mu_\text{III}(E) - 2\mathbb{E} \big[ (1_{E_0} \circ S)_T\big]\mathbb{E} \big[ (u \circ S)_T\big] = 2\mathbb{E} \big[ 1_F(S_T - S_t)\big]\mathbb{E} \big[ (u \circ S)_T\big] = 0.
		$$

	\end{itemize}
Consequently, for each $t$ and $\mathbb P$-null set $F \in \mathcal F_t$ we have 
$$
m_t(F) = \mu((t,T]\times F) = \mu_1((t,T]\times F) + \mu_2((t,T]\times F) + \mu_3((t,T]\times F) = 0,
$$
hence $m_t \ll \mathbb{P} \big |_{\mathcal{F} _t}$.

For every index $t$, we denote $ A _t $ as the \RadonNikodym derivative of $m_t$ with respect to $\mathbb{P} \big |_{\mathcal{F} _t}$. 
For every $\mathcal F_T$ -measurable set $F$, we have
\begin{align*}
\mathbb E [1_F A_T^+]<\infty
\\
\mathbb E [1_F A_T^-]<\infty
\end{align*}
since the finiteness of $\mu$ and the Hahn-Jordan decomposition of $A$, and hence 
$$
\sup \Big(  \sum_i\left | A_{t_{i+1}}(\omega) - A_{t_{i}}(\omega)\right | \Big)<\infty \text{ for a.a. } \omega
$$
where the supremum is taken over the possible partitions of $[0,T]$. 
Consequently, the process $A$ is an integrable process of finite variation.
To see that $A$ is right-continuous, we should notice 
\begin{align*}
&\left \{      \inf_{r>t}A_r^+\in [c,\infty)        \right \}
\\
 = &\bigcap _{r>t}\left \{ A_r^+\in [c,\infty) \right\} \in \mathcal F_{t+} = \mathcal F_t.
\end{align*}
Let $\{t_n\}_n\downarrow t$ with $H\in \mathcal F_t$(hence $(t_n,T]\times H \uparrow (t,T]\times H$), and then we have
$$
\mathbb E\left[1_H A_t^+\right] = \mu^+ ((t,T]\times H) = \lim _{t_n \downarrow t} m_{t_n}^+(H) = \lim_{t_n \downarrow t} \mathbb E \left[1_HA_{t_n}^+\right] = \mathbb E\left[ \lim_n \left(1_HA_{t_n}^+\right)\right] = \mathbb{E} \left[1_H A_{t+}^+\right].
$$
The process $A$ has a right-continuous version since $\sup_{r>t}A^-_r$ is handled in a same manner.
Living in the natural filtration generalized by Brownian motion $B$, for any arbitrary set $H \in \mathcal F_t$, we have a representation
$$
H = \{\omega|\left(B_{t_1}(\omega), B_{t_2}(\omega),\dots,B_{t_n}(\omega),\dots\right)\in G\}
$$
where $\{t_i\}_i$ is a sequence of countable points with $t_i\leq t$ and $G \in \mathcal B (\mathbb{ R ^\infty})$ with $\mathcal B (\mathbb{ R ^\infty})$ given by $\otimes_{i\in \mathbb N} \mathcal B (\mathbb R)$. 
Here we let $H_n$ be the projection of $H$ in head-most coordinates $(t_1, t_2,\dots,t_n)$, and we have $(t_n',T]\times H_n \downarrow (t,T]\times H$ with $t_n':=\sup_{i\leq n} t_i$ and $H_n \in \mathcal F_{t_n'}$. 
The process $A$ is left-continuous since 
$$
\mathbb E\left[1_H A_t^*\right] = \mu^* ((t,T]\times H) = \lim _n m_{t_n'}^*(H_n) = \lim_n \mathbb E \left[1_{H_n}A_{t_n'}^*\right] = \mathbb E\left[ \lim_n \left(1_{H_n}A_{t_n'}^*\right)\right] = \mathbb{E} \left[1_H A_{t-}^*\right]
$$
by the continuity of $\mu$ from above for $* = +,-$. 
Consequently, we can say that $A$ is an integrable continuous process of finite variation.
We can see that
\begin{align*}
\mu((s,t]\times H) = \mathbb E\left[1_HA_s - 1_HA_t \right] = \mathbb E\left[ - \int_s^t 1_H(\omega) \myd A_\tau(\omega)\right] = \mathbb E\left[-\int_0^T 1_E(\tau,\omega)\myd A_\tau (\omega) \right].
\end{align*}

If there is another process $A'$ satisfying the equation~(\ref{fml:representation_of_RGM}), $A'$ is distinguishable from $A$ since the measure $\mu_{A'}$ introduced by $A$ is identical to $\mu$.

With $C :=-\frac{1}{2}A$ and $c(t,\omega)$ as the \RadonNikodym derivative of $\{C(t,\omega)\}_t$ with respect to Lebesgue measure, we complete the proof.
\end{proof}

\subsection{Multi-Dimension Case: Separation of the Terminal Variation}
When it comes to a generalized multi-asset situation, we should stress the benefits of why we consider investment value processes in equation~(\ref{eq:investment}) instead of a self-financing portfolio usually characterized by a linear stochastic differential equation.
Here are the reasons:
\begin{itemize}
	\item 
	With $d = 1$, the self-financing portfolio acts only in the form $X^u_t = u_0S_t$ where the policy $u_0$ is a constant determined by the investor's initial wealth, instead of the generalized form $X^u = (u\circ S)_t$. 
	Due to money limits, it is difficult to examine the evolution of risk contribution in this constant policy situation using the method described in Theorem~\ref{thm:marginal_risk_measure_1d} and Theorem~\ref{thm:representation_of_RGM}.
	\item 
	With $d>1$, people often prefer a controlled linear stochastic differential equation to character the value process of the self-financing portfolio where control is usually a $(d-1)$-dim process. 
	But as one may expect, the (marginal) risk contribution should be a  $d$-dim process instead of $(d-1)$-dim. 
	And the Riesz representation mentioned in Remark~\ref{rmk:marginal_as_dual} can be viewed as macro description of this point. 
\end{itemize}

A particular case where there is one zero-interest risk-free asset and $d-1$ stocks is of particular interest. 
We should notice that these total $d$ assets and the $d-1$ stocks have the same risk. 
It sounds weird, but that is true in the view of risk.
We cannot assign a risk budget to this coupon with no risk.
However, when we throw light on the risk of these stocks, the risk budgets on the stocks imply the allocation of money that comes from the zero-interest coupon.
In practice, it is proper to do this when facing the risk allocation of risky assets.
The volatility-managed portfolio in \cite{moreira2017volatility} is constructed in this manner where they only focus on the risk of risky assets.
We will give an interpretation to it in the view of the risk budget in Subsection~\ref{Section: JF}.

Compared to the classical single-period case, risk contribution is expected to be similar to the covariance matrix structure where the off-diagonal elements make sense. 
The following result can be considered an extension of Theorem~\ref{thm:marginal_risk_measure_1d} regarding the interaction of assets, as it also provides a characterization of the covariance-like structure.

\begin{myTheorem}[Risk Gradient Measure - Generalized Version]
	\label{thm:risk_gradient_measure*}
	For the multi-assets investing case, its terminal variance can be uniquely expressed as 
\begin{equation}
\label{fml:risk_gradient_measure*}
	\text{Var}(X_T^u) = \dfrac{1}{2}\sum _{i = 1}^d\int_{[0,T]\times \Omega}  u^{(i)}(t, \omega) \mu^{i}(\myd t, \myd \omega) + \dfrac{1}{2}\sum_{i\neq j} \int_{[0,T]\times \Omega}u^{(i)}(t, \omega)\eta^{i,j} (\myd t, \myd \omega)
\end{equation}
	where $\mu^{i}$ is the marginal risk measure of $i$-th asset $S^{(i)}$ mentioned in Theorem~\ref{thm:marginal_risk_measure_1d} and $\eta^{i,j}$ is a new signed measure related to mutual effect between $i$-th and $j$-th assets for each $i,j$, or briefly 
	$$
	\text{Var}(X_T^u) =\dfrac{1}{2} \int_{[0,T]\times \Omega}\langle u(t, \omega),{\mu}(\myd t, \myd \omega)\rangle
	$$
	where $\mu$ is the multi-dim measure defined by
	$$
	\mu:= \left[\mu^1 + \sum_{j\neq 1}\eta^{1,j}, \cdots, \mu^d + \sum_{j\neq d}\eta^{d,j}  \right].
	$$
\end{myTheorem}
\begin{proof}
	Noticing $u^\top\circ S = \sum _{i = 1}^d u^{(i)}\circ S^{(i)}$, we can write
	\begin{align*}
	&\text{Var}((u^\top\circ S)_T) \\
	=& \mathbb E\Bigg[\sum_{i = 1}^d\underbrace{\int_0^T ( (u ^{(i)}\circ S^{(i)})_t u^{(i)}_t\myd S^{(i)}_t}_{\text{I-}(i,i)} 
	+ \sum_{i\neq j}\underbrace{\int_0^T (u ^{(i)}\circ S^{(i)})_t u^{(j)}_t\myd S^{(j)}_t}_{\text{I-}(i,j)} 
	\\
	&+  \sum_{i = 1}^d\underbrace{ \int_0^T {u^{(i)}}^2_t\myd \langle  S^{(i)}\rangle_t}_{\text{II-}(i,i)} 
	+ \sum_{i \neq j}\underbrace{ \int_0^T u^{(i)}_tu^{(j)}_t\myd \langle  S^{(i)}, S^{(j)}\rangle_t}_{\text{II-}(i,j)} 
	\\
	&- \sum_{i = 1}^d\underbrace{ \int_0^T u^{(i)}_t\myd S^{(i)}_t\int_\Omega (u^{(i)}\circ S^{(i)})_T(\omega')\mathbb P(\myd \omega')}_{\text{III-}(i,i)} - \sum_{i \neq j}\underbrace{ \int_0^T u^{(i)}_t\myd S^{(i)}_t\int_\Omega (u^{(j)}\circ S^{(j)})_T(\omega')\mathbb P(\myd \omega')}_{\text{III-}(i,j)}\Bigg]
	\\
	=& \sum _{i = 1}^d \dfrac{1}{2}\int_{[0,T]\times \Omega} u(t, \omega) \mu^{i}(\myd t, \myd \omega) + \mathbb E \Bigg[\sum_{i\neq j}\underbrace{\int_0^T (u ^{(i)}\circ S^{(i)})_t u^{(j)}_t\myd S^{(j)}_t}_{\text{I-}(i,j)} 
	\\
	&+ \sum_{i \neq j}\underbrace{ \int_0^T u^{(i)}_tu^{(j)}_t\myd \langle  S^{(i)}, S^{(j)}\rangle_t}_{\text{II-}(i,j)}  - \sum_{i \neq j}\underbrace{ \int_0^T u^{(i)}_t\myd S^{(i)}_t\int_\Omega (u^{(j)}\circ S^{(j)})_T(\omega')\mathbb P(\myd \omega')}_{\text{III-}(i,j)} \Bigg ]
	\end{align*}
	by the {\Ito}'s formula and Theorem~\ref{thm:marginal_risk_measure_1d}. 
	Firstly for convenience, we give some notations of those parts by taking these operators
\begin{align*}
&  \Psi^{i,j}_{\text{I}}u:= \int_0^T (u ^{(i)}\circ S^{(i)})_t u^{(j)}_t\myd S^{(j)}_t,
\\
&
\Psi^{i,j}_{\text{II}}u:=\int_0^T u^{(i)}_tu^{(j)}_t\myd \langle  S^{(i)}, S^{(j)}\rangle_t,
\\
&\Psi^{i,j}_{\text{III}}u:=\int_0^T u^{(i)}_t\myd S^{(i)}_t\int_\Omega (u^{(j)}\circ S^{(j)})_T(\omega')\mathbb P(\myd \omega').
\end{align*}
Hence we can give their related set functions $\myd \Psi_{x}(u;1_E), (x=\text{I,II,III})$ induced by the {\Gateaux} differential manner where we use the notation $1_E(t,\omega) = [1_{E^{(1)}}(t,\omega), \dots, 1_{E^{(d)}}(t,\omega)]^\top$ with $E = \times_{i = 1}^d E^{(i)}$ to treat multi-dim case.
	
Similarly to the Theorem~\ref{thm:risk_gradient_measure*}, the signed measure $\eta_x^{i,j}$ on $\Sigma_p$ induced by $\myd \Psi_{x}^{i,j}(u;v)$ reads
$$
\eta^{i,j}_x(E^{(i)}):= \mathbb E[\myd \Psi_{x}^{i,j}(u;1_E)].
$$
After aggregating the measures $\eta_{x}^{i,j}$($x = \text{I, II, III}$) with fixed pair $(i,j)$, we can get the associated signed measure $\eta^{i,j}:=\sum_x \eta_x^{i,j} $ on $\Sigma_p$. Then for every pair $(i,j)$ we have 
\begin{align*}
&\int _{[0,T]\times \Omega} v^{(i)}(t,\omega)\eta^{i,j}(\myd t, \myd \omega) \Big |_{v=u}
\\
= &\mathbb{ E } \Big[  \underbrace{\int_0^T ( 2(u ^{(i)}\circ S^{(i)})_t u^{(j)}_t\myd S^{(j)}_t}_{\text{I-}(i,j)} 
 + \underbrace{ \int_0^T u^{(i)}_t u^{(j)}_t\myd \langle  S^{(i)}, S^{(j)}\rangle_t}_{\text{II-}(i,j)}
 \\
 &+ \underbrace{ \int_0^T u^{(i)}_t\myd S^{(i)}_t\int_\Omega (u^{(j)}\circ S^{(j)})_T(\omega')\mathbb P(\myd \omega')}_{\text{III-}(i,j)} \Big].
\end{align*}
Finally, with the aggregation of measures $\mu^i$ and $\eta ^{i,j}$ we get (\ref{fml:risk_gradient_measure*}).

The uniqueness of the expression in (\ref{fml:risk_gradient_measure*}) is ensured by taking  arbitrary indicator processes $1_E$ over the $\sigma$-algebra $\Sigma_p$ in a similar manner to that in~Theorem \ref{thm:risk_gradient_measure*}.
\end{proof}

Comparing the result above to the single-period case in Definition~\ref{def:S-Contribution} where 
$$
\dfrac{\partial \rho(v^\top r)}{\partial v}\Big |_{v = u} = 2\Sigma u
$$
gives the structure of the single-period marginal risk contribution, we can say that the $d$-dim measure $\mu$ represent the properties of $\Sigma u$ somehow.
Respectively, $\sigma _{i,j}u_j $ is associated to $\eta ^{i,j}$. 
Actually, whenever $S$ is a multi-dim martingale, the associated product measure $\mu (E)$ can be calculated as  
$$
\mu (E) = 2\mathbb E\Big[ \int_0^T 1_{E}\odot (\sigma_t \sigma^\top_t u_t) \myd t \Big] = 2\mathbb E \Big[\int_0^T 1_E \odot \big \langle u_t, \myd \langle S\rangle_t\big \rangle \Big]
$$ 
for any predictable set $E$.

\begin{myCorollary}[Covariance]
	\label{coro:covariance}
 The terminal covariance of two portfolios $X^u$ and $X^v$ can be calculated as
 $$
 \mathrm{Cov}(X^u_T,X^v_T) = \dfrac{1}{2}\int_{[0,T]\times \Omega}v(t,\omega)\mu^u(\myd t,\myd \omega) = \dfrac{1}{2}\int_{[0,T]\times \Omega}u(t,\omega)\mu^v(\myd t,\myd \omega).
 $$
\end{myCorollary}
\begin{proof}
	Noticing the symmetry of bilinear mapping $\text{Cov}(\cdot, \cdot)$, we can easily get the result in the manners same to that in Theorem~\ref{thm:marginal_risk_measure_1d} and Theorem~\ref{thm:risk_gradient_measure*}.
\end{proof}

\begin{myCorollary}[Linearity of Marginal Contribution]
	\label{coro:RGM_linearity}
	The marginal risk measure $ \mu ^u$ induced by policy $u$ in Theorem~\ref{thm:risk_gradient_measure*} is linear in $u$.
\end{myCorollary}
\begin{proof}
	Only noticing the oprators $\myd \Phi(u;v)$ and $\myd \Psi(u;v)$ are both linear to the policy  $u$ for arbitrary directions $v$, we can get the assertion above.
\end{proof}

\begin{myCorollary}
	\label{coro:representation_of_RGM*}
	With the notation $1_E(t,\omega) = [1_{E^{(1)}}(t,\omega), \dots, 1_{E^{(d)}}(t,\omega)]^\top$, the marginal risk measure $ \mu$ in Theorem~\ref{thm:risk_gradient_measure*} can be uniquely represented as the vector integration
	$$
	 \mu (E) = \mathbb{E} \int_0^T 1_E (t,\omega)\odot  \myd  C(t,\omega) = \mathbb{E} \int_0^T 1_E(t,\omega) \odot c(t,\omega)\myd t
	$$
	by a continuous adapted process $C$.
	And $ c$ is the \RadonNikodym derivative of $ C$ with respect to Lebesgue measure $t$.
\end{myCorollary}
\begin{proof}
	Taking $E^{(1)} = (t,T]\times H$, the section set of $E$, we can check that 
	$$
	m_t^1(H):= \dfrac{1}{2} \mu \Big(\big((t,T]\times H\big) \times \underbrace{\varnothing \times \cdots \times \varnothing }_{d-1\text{ times}}\Big )
	$$
	is a signed measure over $\mathcal F_t$ absolutely continuous with respect to $\mathbb P\big |_{\mathcal F_t}$. Hence we can write $C^{(1)}_t$ the \RadonNikodym derivative of $m_t^1$. Varying $t$ over time interval $[0,T]$, we can get a process $C^{(1)}$ by what we have done in Theorem~\ref{thm:representation_of_RGM}. Similarly $C^{(i)}$ for $ i = 1,\dots, d$ are calculated in a same manner. Finally we put them together
	$$
	C:=[C^{(1)},\cdots, C^{(d)}]^\top
	$$
	and get our result.
\end{proof}

The significance of the formula above is that it enables us to investigate the structure of terminal variance and the aggregated nature of instantaneous risk contributions.
Corollary~\ref{coro:representation_of_RGM*} gives the density of the marginal risk measure in the multi-asset case.
Finally, with the help of this corollary, we can formalize the concept of continuous-time risk contribution, which is the focus of this section.
\begin{myDefinition}[(Marginal) Risk Contribution - Continuous-time]
	\label{def:Continuous-time Contribution}
	Given a $\mathbb R ^d$-valued policy $u$ within the associated density process $c$ in Corollary~\ref{coro:representation_of_RGM*},
	we can define ${c}^{(i)}(t,\omega)$ as the marginal risk contribution of $i$-th asset at time $t$ in the case $\omega$.
	And $u^{(i)}(t, \omega)c^{(i)}(t,\omega)$ is defined as the risk contribution.
	
	Particularly, the policy $u$ is said to be risk-parity if the associated risk distribution satisfies
	\begin{equation}
	\label{fml:C-ERCP}
	u^{(i)}(t,\omega) c^{(i)}(t,\omega) = \lambda\text{ for a.a. }(t,\omega).
	\end{equation}
	In other words,  $u\odot c$ is a constant vector valued process $\lambda e_d$ with some constant $\lambda > 0$.
\end{myDefinition}	
The risk marginal risk contribution $c$ is derived from the {\Gateaux} differential of the terminal variance.
The definition of the marginal risk contribution here is identical to the definition in \cite{cherny_two_2011}, where it goes by the name of  directional risk contribution. 
The risk of an investment $X^u$ can be represented as 
\begin{equation}
\label{fml:C-Variance}
\text{Var}(X_T^u) = \mathbb{ E } \int_0^T u_t^\top  c_t\myd t = \myangle{u,c}
\end{equation}
implying that the total risk is accumulated continuously on the support of $u$. 
Clearly, this accumulation of risk contributions is described uniquely by Euler's homogeneous decomposition, which appears in the statement (\ref{fml:EularDecomposition}) as a continuous version.

\begin{myRemark}
	The critical point of this theorem is the existence of a density process $c$ for the marginal risk measure $\mu$.
	That means the total risk is continuously accumulated over $\Omega \times [0,T]$ within the density $c$.
	Another critical point of this flow representation is that the transition from marginal risk measure to the continuous process $C$ can make the risk contribution we shall define in Definition~\ref{def:Continuous-time Contribution} predictable.
	If we focus on the existence some $c$ satisfying
	\[\Var (u \circ S)_T = \mathbb {E} \int_0^T u_t c_t \myd t, \]
	one may be confused with this theorem because the process $c_t = 2b_t(u\circ S)_t - u_tb_t \mathbb {E}(u\circ S)_T + u_t\sigma_t^2$ can be easily get from 
	\begin{equation}
	\label{eq:marginal-given-by-itos}
	\Var (u\circ S)_T =  \mathbb {E}\int_0^T u_t \left[ 2b_t(u\circ S)_t - u_tb_t \mathbb {E}(u\circ S)_T + u_t\sigma_t^2\right] \myd t
	\end{equation}
	by {\Ito}'s formula.
	Of course, this result of marginal contribution is correct in this situation, but we also need to be alert to the dangers this manner poses.
	When there are some jumps in the dynamics of assets, the marginal process $c$ can not be given by {\Ito}'s formula in equation~(\ref{eq:marginal-given-by-itos}).
	To be able to understand this point better, we can consider the variance in the single-period case
	\[
	\begin{aligned}
	\Var(w^\top r) &= \mathbb {E}\left[w^\top r - \mathbb {E}(w^\top r)\right]^2
	\\
	&= \mathbb {E}\left[\myangle{w,\mathbb {E}(rr^\top w) - r \mathbb {E}(r^\top w)}\right]
	\\
	&= \myangle{w,\Lambda w}.
	\end{aligned}   
	\]
	If we look only at the results, both the terms $\mathbb {E}(rr^\top w) - r \mathbb {E}(r^\top w)$ and $\Lambda w$ can be considered as the marginal risk contribution.
	It is not hard to check $\mathbb {E}(rr^\top w) - r \mathbb {E}(r^\top w) + \xi$ where $\xi$ is a zero-mean random variable also meets the equation above.
	Taking the marginal risk measure can avoid this disagreement in the expression.
	The first one $\Lambda w$ is accepted for its predictable property.
	So it is necessary and reasonable to take the marginal risk measure and the associated density process $c$ in Theorem~\ref{thm:representation_of_RGM}.
	In particular, when dealing with risk budgeting problems in Section~\ref{Section: Budgeting}, an ideal form of risk contribution should coincide with a given predictable risk budget. 
	Hence, this flow representation appears to be of importance.
\end{myRemark}

\begin{myCorollary}[Continuity of the Marginal Risk Contribution on Policy]
	\label{coro: Continuity of Contribution wrt Policy}
	Suppose we have two policies $u$ and $u'$ with the notation $\delta u := u-u'$. 
	Then we have an estimation of the difference between two marginal risk contributions
	$$
	\mathbb E \Bigg[ \int_0^T 1_E^\top c_t^{\delta u}\myd t \Bigg]\leq K\big \lVert \delta u \big \rVert_{\mathbb S^\infty}
	$$
	for arbitrary predictable sets $E$.
\end{myCorollary}
\begin{proof}
	For convenience, we provide the proof here for the case $d = 1$, and the left part for the multi-dim case can be easily derived in the same manner.
	Starting with three {\Gateaux} differential originally appearing in Lemma~\ref{lm:gateaux_differential}, we have
		\begin{align*}
		\mathbb{E} [\myd \Phi_{\text{I}}(u;v)] & \leq 2\mathbb{E}\Bigg [\sup_{s\leq T} \big |u_s\big |\int_0^T\Big (S_tv_t+(v\circ S)_t\Big) \myd S_t \Bigg]\\
		& \leq 2\big \lVert u \big \rVert_{\mathbb S^\infty}  \big \lVert v \big \rVert_{\mathbb S^\infty} \Bigg \lVert \int_{ 0 }^T 2S_t\myd S_t \Bigg \rVert_{L^1},\\
		\mathbb{E} [\myd \Phi_{\text{II}}(u;v)] & \leq 2\mathbb{E} \Big  [\sup_{s\leq T} \big |u_s\big | \sup_{s\leq T} \big |v_s\big | \langle S \rangle_T\Big ]\\
		&= 2\big \lVert u \big \rVert_{\mathbb S^\infty}  \big \lVert v \big \rVert_{\mathbb S^\infty} \Big \lVert \big \langle S \big \rangle_T^{1/2} \Big \rVert_{L^2},\\
		\mathbb{E} [\myd \Phi_{\text{III}}(u;v)] & \leq 2 \big \lVert u \big \rVert_{\mathbb S^\infty}  \big \lVert v \big \rVert_{\mathbb S^\infty} \mathbb{E}^2\big [\big |S_T\big |\big ].
		\end{align*}
	Finally replacing $u$ by $\delta u$ and $v$ by $1_E$, with the representation in Theorem~\ref{thm:risk_gradient_measure*} we get the result 
		$$
		\mathbb E \Big[ \int_0^T 1_E c_t^{\delta u}\myd t \Big]\leq K\lVert \delta u \rVert_{\mathbb S^\infty}, \forall E \in \Sigma_p.
		$$
\end{proof}

\begin{myCorollary}[Risk Contributions of Self-financing Portfolios]
	\label{coro: explicit risk contribution}
	 Given that the dynamic of assets processes are characterized by a multi-dim SDE
\[
	\begin{cases}
		\myd S_t = \mathrm{Diag}(S_t)(b_t \myd  t + \sigma_t \myd W_t)\\
		S_0 = s_0,
	\end{cases}\]
	we can give the marginal risk contribution
	\[
	c_t = 2\mathrm{Diag}(S_t)b_t u_t^\top S_t + \mathrm{Diag}(S_t)\sigma_t \sigma_t^\top  \mathrm{Diag}(S_t)u_t - \mathrm{Diag}(S_t)b_t(\mathbb E X_T^u + x_0).
	\] 
	Moreover, working on the self-financing case, i.e., the portfolio $X^u$ with properties (\ref{eq:investment}) and (\ref{eq:portfolio}),  the related risk contribution of the policy $u$ is given by
	\begin{itemize}
		\item (in a general share manner) $$(u\odot c)_t = \mathrm{Diag}(u_t) c_t,$$
		\item (in a money manner where policies are specified by the money $M$ allocated on each assets)
		$$
		(u\odot c)_t = 2\mathrm{Diag}(M_t)b_t X_t^u + \mathrm{Diag}(M_t) \sigma _t \sigma _t^\top M_t- \mathrm{Diag}(M_t) b_t (\mathbb{ E} X_T^u + x_0),
		$$
		\item (in a weight manner where policies are specified by the weights $w$ allocated on each assets)
		$$
		(u\odot c)_t = 2 \mathrm{Diag}(w_t) b_t X_t^2 + \mathrm{Diag}(w_t) \sigma_t \sigma_t ^\top w_t X_t^2 - \mathrm{Diag}(w_t)b_t X_t (\mathbb E X_T + x_0).
		$$
	\end{itemize}
\end{myCorollary}

\section{Risk Budgeting Problem}
\label{Section: Budgeting}
Compared to the classical result in Example~\ref{example:single-period}, the risk contribution and marginal risk contribution in the continuous-time case depends not only on the assets but also on the time $t \in [0,T]$ and the sample points $\omega \in \Omega$.
These latter two newly appearing terms are of our interest for risk allocation over $\Omega \times [0,T]$.
Having depicted the structure of terminal variance in Section~\ref{Section: Contribution}, we are facing the risk budgeting problem. 
We write $\mathcal D$ the collection of risk distributions induced by all feasible policies in $\mathbb S^\infty$ under the assets in a fixed market.
Within fixed assets, we write $\mathcal D$ the collection of risk distributions induced by all possible policies in $\mathbb S^\infty$. 
The elements in $\mathcal D$ can be obtained through the implementation of some policies.

\begin{myProblem}[Risk Budgeting - Continuous-time]
	\label{prob:C-budget}
	With the terminal variance acting as the risk measure of investments and a continuous $\mathbb R^d$-valued process $\beta \in \mathcal D$, the modelers seek to find suitable policies $u^\star$ satisfying
	$$
	u^\star_t \odot c^{u^\star}_t = \beta _t
	$$
	for each $t$.
\end{myProblem}

To obtain the policy associated with the risk budget, we can design some `loss' functions. 
To budget the standard deviation or the variance in  the single-period case, one can take the loss function like
\begin{equation}
\label{eq:residual-sum-squares-single-period}
J(u) = \sum_{i=1}^d\left(u_i\dfrac{\partial \rho(u)}{\partial u_i} - \beta _i\right)^2
\end{equation}
to make the risk contribution near enough to the budget $\beta$.
However the loss functions in a penalty form (for example $J(u) = \mathbb {E}\int_0^T \lVert u_t\odot c_t - \beta_t \rVert ^2\myd t$) can not work well in the continuous-time case, because it spends too much on the calculation of the marginal term $(t,\omega)$-pointwisely.
The result of \cite{maillard2010properties} gives an optimization to obtain the long-only risk parity policy of the variance where the marginal risk contribution does not need to appear explicitly. 
The optimization is a minimization of the variance with the constraints on the weights given by
\begin{equation}
\label{eq:Maillard-optimization}
    \begin{aligned}
        \text{to minimize }& J(x) = x^\top \Lambda x
        \\
        \text{s.t. }& \begin{cases}
            \sum_{i = 1}^d \log x_i \geq C
            \\
            x_i >0, \forall i   .
        \end{cases}
    \end{aligned}
\end{equation}
Then the risk parity weight $w^\star$ is given by the regularization of $x^\star$, i.e., $w^\star = \dfrac{x^\star}{\sum_i x_i}$. We adopt the loss function in this form where the minimization of the variance and the constants of the shares/weight are concerned and develope it in the continuous-time case.

Inspired by the tradeoff between the diversification of policies and the minimization of risk, we modify the loss function in (\ref{eq:Maillard-optimization}) as
\begin{equation}
\label{eq:modified-Mailard-optimization}
J(u) = \sum_{i = 1}^d -\beta_i\log u_i + u^\top \Lambda u
\end{equation}
where $\beta>0$ is the exogenously risk budget.
Here the vector $u$ is the share of assets and the matrix $\Lambda$ is the covariance matrix of the prices instead of returns. Hence the risk budget $\beta$ should be understood as the budget according to the price rather than the yield rate.
\subsection{Relations between Risk-based Policy and Terminal Variation in Continuous-time Case}
As a generalization of the optimization in the equation~(\ref{eq:modified-Mailard-optimization}) the connection between the continuous-time policy with risk budget process $\beta$ and the optimization problem
\begin{equation}\label{fml:C-Optimization}
	\text{to minimize }  J(u) = \mathbb{ E } \Big[ \int_0^T -\sum_{i = 1}^d\beta^{(i)}_t\log u^{(i)}_t \myd t \Big] + \text{Var}[(u\circ S)_T]
\end{equation}
over possible policies $u \in \mathbb S ^\infty$ remains to be explored.

\begin{myTheorem}[Continuous-time Risk Budgeting]
	\label{thm:C-Optimization}
	Suppose the non-degeneration condition  in Definition~\ref{def:ND-0} is satisfied.
	Given a positive $\mathbb R_+ ^d$-valued process $\beta \in \mathcal D$,  the risk budgeting optimization~(\ref{fml:C-Optimization}) has a unique solution $u^*$ in $\mathbb S^\infty$. And the associated risk contribution of the optimal policy $u^*$ meets the pre-given risk budget $\beta$. 
\end{myTheorem}
\begin{proof}
	The positive process $\beta$ ensures the strict convexity of the first term $u\mapsto \mathbb{ E } \Big[ \int_0^T -\sum_{i = 1}^d\beta^{(i)}_t\log u^{(i)}_t \myd t \Big]$.
    As for the terminal variance, it is naturally convex as a risk measure.
    Hence $J$ is a convex functional.
    To reach the minimum,  we need to check the closedness and the recession direction of $J$.
    The closedness is guaranteed by the space $\mathcal {H}_S^2$ of assets and the space $\mathbb {S}^\infty$ of policies.
    To see the recession functioanl $(J0^+)(v):=\sup_{\theta >0}\dfrac{J(u+\theta v) - J(u)}{\theta}$ of $J$, we notice the slope function of $J$ satisfies
    \[
    \begin{aligned}
        &\dfrac{J(u+\theta v) - J(u)}{\theta}
         \\
         =&\left[\mathbb {E}\int_0^T -\sum_{i = 1}^d\dfrac{\beta_t^{(i)}(\log (u_t^{(i)}+v_t^{(i)}) - \log (u_t^{(i)})}{\theta } + \Cov(X_T^u, X_T^v) + \theta \Var (X_T^v)\right]
         \\
         \geq & \left[\mathbb {E}\int_0^T -\sum_{i = 1}^d\dfrac{\beta_t^{(i)}v_t^{(i)}}{u_t^{(i)}} + \Cov(X_T^u, X_T^v) + \theta \Var (X_T^v)\right]
    \end{aligned}
    \]
    for any point $u\in \mathbb {S}^\infty$.
    On the directions with $\Var (X_T^v)>0$, the recession $(J0^+)(v) = \infty$. 
    But on the directions with $\Var(X_T^v) = 0$, the recession functional $J0^+$ may be less equal to zero(depends on $\Cov(X_T^u, X_T^v)$).
    The non-degenerate condition in Definition~\ref{def:ND-0} guarantees $\Var(X_T^v)>0$ for any direction $v\neq0$.
    Hence what left us to see is the local minimum of $J$.

    By the slope function above and Corollary~\ref{coro:covariance}, we can see the first order condition of the functional $J(\cdot)$ at the local minimum $u^\star$ which satisfies\footnote{For two non-zero $\mathbb R^d$vectors $a$ and $b$, $a\oslash b = \left [\dfrac{a^{(1)}}{b^{(1)}}, \dots, \dfrac{a^{(d)}}{b^{(d)}}\right ]^\top$.}
        $$
        \lim_{\theta\downarrow 0}\dfrac{J(u^\star + \theta v) - J(u^\star)}{\theta} = \mathbb E \Big[\int_{[0,T]\times \Omega} -\beta(t,\omega)\odot v(t,\omega)\oslash u^\star(t,\omega)+ v(t,\omega)\odot c^{u^\star}(t,\omega) \myd t\Big] = 0
        $$
    for arbitrary policies $v$. Hence we have
        $$
        -\beta^{(i)}(t,\omega) + {u^\star}^{(i)}(t,\omega){c^{u^\star}}^{(i)}(t,\omega) = 0,\text{ for a.a. }(t,\omega)\text{ in }\Sigma_p \text{ and } i = 1,\dots,d
        $$
    and the equation above also implies
    \[\mathbb {E}\int_0^T \sum_{i = 1}^d\lVert u^\star_t\odot c_t^{u^\star} - \beta _t \rVert^2\myd t=0.\]
    
    The level sets $\lev _\alpha J=\{u | J(u)\leq \alpha\}, \alpha\in \mathbb {R}$ are certain bounded closed convex sets and also for $\alpha=\inf J$. Now suppose $u$ and $u'$ are two differnet policies in the minimum set $\lev _{\inf J}(J)$. 
    We can take $v = u'-u$ and the policies $u + \theta v$ are also in the minimum set.
    The slope function between $u$ and $u+\theta v$ shows
    \[\dfrac{J(u+\theta v) - J(u)}{\theta} = \mathbb {E}\int_0^T -\sum_{i = 1}^d\dfrac{\beta_t^{(i)}v_t^{(i)}}{u_t^{(i)}} + \Cov(X_T^u, X_T^v) + \theta \Var (X_T^v) >0\]
    which implies a contradiction.
    Hence the optimal solution is unique.
\end{proof}

Now we want to make an interpretation of the term $\mathbb{ E } \Big[ \int_0^T -\sum_{i = 1}^d\beta^{(i)}_t\log u^{(i)}_t \myd t \Big]$.
Actually, the result of \cite{maillard2010properties} takes the optimization~\ref{eq:Maillard-optimization} as alternative to the residual sum of squares in \ref{eq:residual-sum-squares-single-period}.
The risk parity portfolio is similar to a minimum variance portfolio subject to a diversification constraint on the weights of its components.
Here we want to give the economic interpretation of the diversification term.
The derivation of the term $\mathbb{ E } \Big[ \int_0^T -\sum_{i = 1}^d\beta^{(i)}_t\log u^{(i)}_t \myd t \Big]$ is given step by step.

First, we can consider a basic single-period optimization with constants on shares, namely
\begin{equation}
    \begin{aligned}
        \text{to minimize }& u^\top \Lambda u
        \\
        \text{s.t. }& u\in \mathcal {C}.
    \end{aligned}
\end{equation}
The aim of the area $\mathcal {C}\subset \mathbb {R}_+^d$ is to prevent the shares from concentrating.
We can appoint linear constraints $u_i \geq a_i>0$ for $i = 1, \cdots, d$.
The optimal policy of this linear constraint has a constant marginal risk contribution by a simple calculation.
The boundary $a$ is exogenously given and usually follows the varying trend of the prices.
In other words, the boundary $a$ may have an exponential growth which implies an exponential boundary of shares.
To avoid this, one can take logarithmic constraints of the shares, i.e.,
\begin{equation}
\mathcal {C} = \{u| \log u_i \geq a_i \text{ for }i = 1,\cdots,d.\}.
\end{equation}
It is easy to check that the risk contribution of the optimal policy happens to be the Lagrangian multipliers of the optimization with the logarithmic constraints.
As a special case, the optimal policy with the constraint $\sum_i \log u_i \geq a $, also the constraint in \cite{maillard2010properties}, shows an equal risk distribution.
Having noticed the connection between the Lagrangian multipliers of the log-constrained optimization and the risk contribution of the associated solution, we can generalize the Lagrangian to the function
\[J(u) = -\sum_{i = 1}^d \beta_i\log u_i + u^\top \Lambda u\]
where the multipliers are replaced by the budget $\beta$.
The optimization solution above can give the policy with the risk budget $\beta$.
Finally, with the help of the results in Section~\ref{Section: Contribution}, we put it in a continuous-times style and get the optimization~(\ref{fml:C-Optimization}).

In perspective on the constraints, the continuous-time optimization problem (\ref{fml:C-Optimization}) can also be considered as a constrained convex minimization of terminal variance, i.e., 
    \begin{align*}
    &\text{minimize } \text{Var}(X_T^u)\\
    &\text{s.t. } u\in  C,
    \end{align*}
where convex constraint area $C$  is given by 
    $$
    C=\{(t,\omega) | \log u^{(i)}(t,\omega)\geq \alpha^{(i)}(t,\omega) \text{ for a.a. }(t,\omega) \text{ in } \Sigma_p \text{ and }i = 1,\dots, d\}.
    $$
This convex constraint area $C$ may be thought of as a regulatory mechanism that prevents the concentration of the shares~$u$ from nearing zero $(t,\omega)$-pointwisely.
When changing the shares~$u(t,\omega)$ for each pair~$(t,\omega)$, the Lagrangian multiplier can be regarded as the tolerance of the exponential growth in the level of $\alpha(t,\omega)$.

Now that the risk has a distribution over $([0,T]\times \Omega)$, we can also think about the constraint term in the way of distribution.
To see this, we can consider the {\Doleans} measures $p_u$ and $ p_{\beta}$ introduced by the multi-dim processes $u$ and the process $\beta$ over $[0,T]\times \Omega$ by 
    \begin{align*}
    p_u^{(i)}(E)& := \mathbb E \int_0^T 1_E u^{(i)}_t\myd t
    \\
    p_{\beta}^{(i)}(E)& := \mathbb{ E }\int_0^T 1_E \beta^{(i)}_t \myd t
    \end{align*}
for each $i$ and predictable set $E$. 
$p_u$ is the distribution of the undetermined shares and $p_\beta$ is the distribution of risk budget.
And then the Kullback-Leibler divergence---also known as the relative entropy---from $p_u$ to $p_{\beta}$ implies
    $$
    D_{\text{KL}}(p_{\beta}||p_u) = \sum_{i = 1}^d \int _{[0,T]\times \Omega}\log\left (\dfrac{\myd p^{(i)}_{\beta}}{\myd p^{(i)}_{u}}\right )\myd p^{(i)}_{\beta}
    =\mathbb E\left [\int _0^T \sum_{i = 1}^d \left( -\beta_t^{(i)}\log u_t^{(i)}  + \beta_t^{(i)} \log \beta_t^{(i)}\right)\myd t\right ]
    $$
which  measures the divergence between $p_u$ and $p_\beta$.
Thus, given a positive risk budget process $\beta$, the objective of reducing both the K-L divergence $D_{\text{KL}}(p_{\beta}||p_u)$ and the terminal variance $\Var (X_T^u)$ is identical to that in (\ref{fml:C-Optimization}).
As a particular example, taking a naive risk budget $\beta(t,\omega)\equiv \lambda$ for some positive constant $\lambda$ and a.a. $(t,\omega)$, this implies the K-L divergence from undetermined $p_u$ to a uniform distributed $p_\lambda$ over $[0,T]\times \Omega$.
Overall, the risk budgeting optimization is equal to balancing the diversification of the policy $u$ and the minimization of the total risk $\Var (X_T^u)$.
\begin{myRemark}
	Theorem~\ref{thm:C-Optimization} discusses the continuous-time risk parity/budgeting cases within a positive share setting in order to establish the connection between admissible set $\mathcal D$ and possible budget processes $\beta$.
Those risk budgeting problems with negative supports $C^-\subset ([0,T]\times \Omega)^d$ where the policies $u(t,\omega)$ are negative-valued remain interesting.
    In light of the regulation of the shares, we can divide the constraint area into two sections
        $$
        \begin{cases}
        -\log u^{(j)}(t,\omega)\leq  \alpha^{(j)}(t,\omega) \\  \text{ for some }(t,\omega)\in C^+ \text{ and }j = 1,\dots, d \text{ with positive budget;}\\
        -\log -u^{(k)}(t',\omega')\leq \alpha^{(k)}(t',\omega') \\ \text{ for some }(t',\omega')\in C^- \text{ and }k = 1,\dots, d \text{ with negative budget.}
        \end{cases}
        $$
     Then we can deduce a generalized optimization
        $$
        \text{minimize }  J(u) = \mathbb{ E } \Big[ \int_0^T -\sum_{i = 1}^d\beta^{(i)}_t\log (\delta^{(i)}_tu^{(i)}_t) \myd t \Big] + \text{Var}[(u\circ S)_T]
        $$
    where the risk budget $\beta$ can be negative on $C^-$ and $\{1,-1\}$-valued sign process $\delta$ coinciding with $\beta$ is also negative on $C^-$. Additionally, we can observe that the policy $u$ and its inverse $-u$ have the same risk distribution.

    The behavior of the optimal policy is also of interest when there is no constraint on some area $C^0$ or the risk budget satisfies $\beta =0 $ on the area $C^0 \in \Sigma_p$. In this case, the optimization~\ref{thm:C-Optimization} still has an optimal solution and the associated policy $u^\star$ also meets the budget $\beta$ for a.a. $(t,\omega)$., i.e., 
    \[u^\star_t\odot c^{u^\star}_t = \beta_t.\]
    The equation above implies that the risk contribution of $u^\star$ is zero on $C^0$ and also implies the marginal risk contribution or the share itself is zero on that zero-risk-budget area.
    The policy $u^\star$ is one of the policies with risk contribution identical to $\beta$ on the complement of $C^0$, and $u^\star$ is also the minimized total variance of these policies.
    The zero-valued marginal risk contribution and the share on $C^0$ make sure that is true.
\end{myRemark}


\subsection{Ramifications}

In simple terms, both the policy and the exogenous risk budget depend on the information, and the information is equal to the measurability of the processes in the stochastic language.
The critical point here is how the risk budgeted policy behaves when given rough information.
When dealing with less information, it can be challenging to establish a policy whose risk contribution is the same as the budget $\beta$. 
Suppose we can only make a choice at the start of $[0,T]$ and can not change it. 
Given the budget $\beta$, how to determine the best policy in this case?
The corollary following gives the relation between the risk budget and the optimal policy with less information.

\begin{myCorollary}[Projection]\label{coro:Projection}
	Given an integrable positive risk budget $\beta$ and the admissible policy set $\mathcal U = \mathbb S^\infty \cap L^0(\Sigma_p')$ where $\Sigma_p'$ is a sub $\sigma$-algebra of $\Sigma_p$, the solution of the optimization
		$$
		\min_{u\in \mathcal U}  J(u) = \mathbb{ E } \Big[ \int_0^T -\sum_{i = 1}^d\beta^{(i)}_t\log u^{(i)}_t \myd t \Big] + \text{Var}[(u\circ S)_T]	
		$$
	gives a policy $u^\star$ of which the associated risk distribution is the projection of $\beta$ onto $\Sigma_p'$.
\end{myCorollary}
\begin{proof}
	The assertion can be quickly get from the first-order condition
		$$
		\mathbb E \Big[\int_{[0,T]\times \Omega} -\beta(t,\omega)\odot 1_E(t,\omega)\oslash u^\star(t,\omega)+ 1_E(t,\omega)\odot c^{u^\star}(t,\omega)    \myd t\Big] = 0
		$$
	with arbitrary $\Sigma_p'$-mesurable set $E$, and hence 
		$$
		-\bar \beta^{(i)}(t,\omega) + {u^\star}^{(i)}(t,\omega){c^{u^\star}}^{(i)}(t,\omega) = 0,\text{ for a.a. }(t,\omega)\text{ in }\Sigma_p' \text{ and } i = 1,\dots,d
		$$
	where $\bar \beta$ is the projection from $\Sigma_p$ to $\Sigma_p'$ within measure $\mathbb P \times \mathrm{Lebesgue}$.
\end{proof}
The projection corollary above also shows that the single-period risk budgeting problem is the degenerate case of continuous-time risk budgeting problems where the risk budgets are constants and the policy is constrained in the naive $\sigma$-algebra.

The optimization (\ref{fml:C-Optimization}), unlike the well-known expected utility problems, is not a pure stochastic control problem.
The embedding lemma is typically employed to solve the continuous-time mean-variance frontier problem where the mean-variance object is converted to a stochastic control problem. 
The lemma following gives a similar method to solve the risk budgeting problem in a stochastic control framework.
\begin{myLemma}[Embedding, original version in \cite{zhou2000continuous}]
\label{lm:embedding}
Suppose $u^\star$ is the optimal policy in Theorem~\ref{thm:C-Optimization}.
Denote $u_\gamma$  the solution of the following parametrized auxiliary problem
	$$
	\text{minimize }  J_\gamma(u) = \mathbb{ E } \Big[ \int_0^T -\sum_{i = 1}^d\beta^{(i)}_t\log u^{(i)}_t \myd t      + \gamma (u\circ S)_T + (u\circ S)_T^2\Big]
	$$
with same condition in Theorem~\ref{thm:C-Optimization}. 
Then we have 
	$$
	u^\star \in \bigcup _{\gamma\in \mathbb R}u_\gamma.
	$$
Moreover, $u^\star$ is identical to $u_{\gamma^\star}$ where $\gamma^\star = -2\mathbb E[(u^\star\circ S)_T]$.
\end{myLemma}
\begin{proof}	
	We can define a function 
		$$
		f(x,y,z) = x-y^2+z
		$$
	which is a concave function with 
		$$
		f\Big(\mathbb E\big[(u\circ S)_T^2\big], \mathbb E\big[(u\circ S)_T\big], \mathbb E\Big[\int_0^T-\sum_{i = 1}^d\log u_t^{(i)}\myd t\Big]\Big ) = J(u).
		$$
	If $u^\star$ is not optimal for $J_{\gamma^\star}$, there will be a $u'$ optimal for $ J_{\gamma^\star}$ satisfying
		\begin{align*}
		\mathbb{ E } \Big[ \int_0^T -\sum_{i = 1}^d\beta^{(i)}_t\log {u^\star} ^{(i)}_t \myd t      + \gamma (u^\star\circ S)_T + (u^\star\circ S)_T^2\Big] 
		\\> \mathbb{ E } \Big[ \int_0^T -\sum_{i = 1}^d\beta^{(i)}_t\log u'^{(i)}_t \myd t      + \gamma (u'\circ S)_T + (u'\circ S)_T^2\Big].
		\end{align*}
	On the one hand, the convexity of $f$ implies
	\begin{align*}
		f&\Big(\mathbb E\big[(u^\star\circ S)_T^2\big], \mathbb E\big[(u^\star\circ S)_T\big], \mathbb E\Big[\int_0^T-\sum_{i = 1}^d\log {u_t^\star}^{(i)}\myd t\Big]\Big ) 
		\\
		&+ 2\mathbb E\big[(u^\star\circ S)_T\big]\Big( \mathbb E\big[(u'\circ S)_T\big] - \mathbb E\big[(u^\star\circ S)_T\big]  \Big)
		\\
		&+\Big( \mathbb E\big[(u'\circ S)_T^2\big] - \mathbb E\big[(u^\star\circ S)_T^2\big] \Big) + \Big(\mathbb E\Big[\int_0^T-\sum_{i = 1}^d\log u_t'^{(i)}\myd t\Big] - \mathbb E\Big[\int_0^T-\sum_{i = 1}^d\log {u^\star_t}^{(i)}\myd t\Big]  \Big)
		\\
		\geq&  	f\Big(\mathbb E\big[(u'\circ S)_T^2\big], \mathbb E\big[(u'\circ S)_T\big], \mathbb E\Big[\int_0^T-\sum_{i = 1}^d\log u_t'^{(i)}\myd t\Big]\Big )
	\end{align*}
which induces a contradiction to the original problem, and hence we get the result $ J(u^\star) \leq J(u')$.
\end{proof}

Lemma~\ref{lm:embedding} implies that risk budgeting problems can be accomplished by solving a family of parameterized stochastic control problems. 
The optimal solution of optimization~\ref{fml:C-Optimization} is embedded in the parameterized solutions of those auxiliary problems.
Usually, these auxiliary problems can be solved identically.
To find the optimal solution, we can check $\gamma^\star = -2\mathbb E[(u^\star\circ S)_T]$.

Additionally, the item $\mathbb E [\gamma (u\circ S)_T ]$ appearing in the auxiliary functional $J_\gamma$ can give a parameterized signed measure 
$$
\mu_\gamma(E) = \mathbb{ E } [\gamma \int_0^T 1_E\myd S_t].
$$
This parameterized measure can be considered as a guess of $\mu_{\text{III}}(E) = \mathbb E[\myd \Phi_{\text{III}}(u;1_E)]$ established in Theorem~\ref{thm:marginal_risk_measure_1d}.
And when we get the right guess $\gamma ^\star$, the related $\hat \mu_{\text{III}}$ meets true $\mu_{\text{III}}$. 
Actually, in the language of risk contribution, this embedding technique in \cite{zhou2000continuous} gives a modification of the marginal risk contribution.

\section{Examples}

It is natural to extend the risk contribution idea from the classical covariance matrix structure to the continuous-time setting. 
This section will provide three examples to demonstrate the concepts of risk contribution and risk budgeting.
\begin{itemize}
	\item In the first subsection, we derive a strategy based on a concrete risk budget. 
    This derived strategy is identical to the volatility-managed portfolio in \cite{moreira2017volatility} where the portfolio is based on the volatility anomaly in the time axis.
    We will derive this strategy in a risk budget manner without the tenet of risk anomaly and give the average return at long-term growth.
	\item The risk contributions of policies depend on the concrete model of the dynamics.
	As for the second example, we will relate the risk contribution to the coefficients in the SABR model, which is commonly used for option modeling.
	The risk contribution can be derived through the options of the underlying assets.
	\item  As a tool to detect the concentration of the risk, the final example gives the risk contribution of the continuous-time mean-variance portfolio in \cite{zhou2000continuous}.	
\end{itemize}
\subsection{Volatility-Managed Portfolios}
\label{Section: JF}
Unlike consensus that higher risk always brings a higher return, low-risk anomaly states that the stocks with low risks have higher returns than the stocks with higher risk.
This smile risk-return curve happens not only in the cross-section but also in the time axis.
The volatility-managed portfolio in \cite{moreira2017volatility} suggests enlarging the investment on the risky assets when their risk is at a lower level.
To be more specific, they investigate the relationship between the lagged volatility factor and the current return of stocks.
Empirical results show that 
\begin{itemize}
	\item the current month's average returns have a weak relationship with the lagged realized volatility;
	\item whereas the current month's volatility has a strong relationship with lagged realized volatility.
\end{itemize}
Consequently, the return-volatility ratio has a negative correlation with the volatility because of the volatility clustering.
In other words, the efficiency of the assets is higher when the volatility is at a lower level.
The difference in the return-volatility ratio in the time axis implies different risk budgets at different times.
The agents can time volatility by shorting risky assets during periods of high volatility and accepting risk during periods of low volatility to gain from the mean-variance trade-off.

How does the volatility-managed portfolio work?
In \cite{moreira2017volatility}, volatility-managed portfolios are designed by scaling monthly returns by the inverse of the previous month's realized variance, therefore reducing risk exposure when the realized variance was at a high level and vice versa.
This managed strategy can be summarized as 
	\begin{equation}
	\label{eq:JF2017}
	f_{t_{i}}^\sigma  = \dfrac{\hat c}{\hat \sigma _{t_{i-1}}^2}f_{t_i}
	\end{equation}
where $f_{t_i}$ is the buy-and-hold portfolio return of current month, $\hat \sigma_{t_{i-1}}^2$ represents the proxy for portfolio's instantaneous variance and constant $\hat c$ controls the average risk exposure of this managed strategy. 
The buy-and-hold portfolio can be a capital-weighted market index or a simple stock.
The investor can borrow money from a bank account (risk-free assets) to scale this buy-and-hold portfolio.
Through this simple scheme, though volatility does not forecast returns, the efficiency of the investment is improved by time-selected volatility.

This scheme implies the idea of the risk budget in the time axis:
the risk budget can be taken at a high level when the volatility is lower.
We will give a continuous-time risk budgeting interpretation of this scheme.
In this subsection, we assume that the dynamic of the underlying asset is in a linear form
	\begin{equation*}
		\begin{cases}
		\myd S_t = r_tS_t \myd t + \sigma_t S_t\myd W_t\\
		S_0 = s
		\end{cases}
	\end{equation*}
where $W$ is a $1$-dim Brownian motion under the physical probability measure $\mathbb P$. 
The coefficients return rate $r$, volatility rate $\sigma$ are locally bounded processes such that this Lipschitz SDE can be well-posed in $\mathcal H^2_{\mathcal S}$ and satisfies the non-degenerate condition in Definition~\ref{def:ND-0}. 
The solution to this SDE can be considered the price of a single stock or the market portfolio, e.g., market stock indices.

The reference probability measure $\mathbb Q$ appearing in budgeting optimization can be identified as the risk-neutral measure constructed by a \Doleans-Dade exponential martingale
	\begin{equation}
	\dfrac{\myd \mathbb Q}{\myd \mathbb P} \Big| _{\mathcal F_t} = \mathcal E(\sigma_t^{-1}b_t\circ W)_t = \mathcal E(\theta \circ W)_t 
	\end{equation}
under which $S$ is a martingale and $\theta_t := \sigma^{-1}_tr_t$ is called the market price of risk. 
Hence under this reference probability $\mathbb Q$, the terminal variance can be calculated as
\begin{equation}
\mathrm{Var}^{\mathbb Q}(X_T^u) = \mathbb E ^{\mathbb Q} [\langle u\circ S \rangle _T] = \mathbb E^{\mathbb Q}\left[\int_0^Tu_t^2S_t^2\sigma_t^2 \myd t \right ].
\end{equation}
Since $S$ is in exponential growth and the terminal variance is homogeneous of order $2$, the parity risk budget $\beta$ we will take should be in a regularized form
\begin{equation*}
\beta _t^{1/2} \propto \dfrac{S_t}{\sigma_t}
\end{equation*} 
for each $t$. 
Without the loss of generality, we can take 
	\begin{equation*}
	\beta_t ^{1/2}=\dfrac{\hat c}{\sigma_t}S_t > 0
	\end{equation*}
where this positive proportion coefficient $\hat c$ plays the same role in the formula~(\ref{eq:JF2017}) to control the total risk level. 

Given this regularized risk parity budget, the optimal condition of the associated convex risk budgeting optimization 
    \begin{equation*}
    J(u) = \mathbb E^{\mathbb Q}\left( \int_0^T-\beta _t \log u_t \myd t \right) + \dfrac{1}{2} \mathrm{Var}^{\mathbb Q}(X_T^u)
    \end{equation*}
is given by
    \begin{equation*}
    -\dfrac{\beta_t}{u^\star_t} + S_t^2 u^\star_t \sigma_t^2=0.
    \end{equation*}
The optimal policy $u^\star$ can be straightly calculated by
    \begin{equation*}
    u^\star _t = \dfrac{\hat c}{\sigma_t ^2}
    \end{equation*}
which is identical to the core equation~(\ref{eq:JF2017}) in \cite{moreira2017volatility} and derives the optimal investment process
    \begin{equation}
    \begin{aligned}
    \myd X_t^\star & = u_t^\star \myd S_t
    \\
    & = \dfrac{\hat c}{\sigma_t^2}S_t (r_t \myd t + \sigma _t \myd W_t).
    \end{aligned}
    \end{equation}
We can rewrite this investment process in the form 
    \begin{equation}
    \myd X_t^\star = \underbrace{\beta_t^{1/2}}_{\text{amplifier}} (\underbrace{\theta _t \myd t}_{\text{signal}} + \underbrace{\myd W_t}_{\text{noise}}).
    \end{equation}
Having been back to our physical world $\mathbb P$, our optimal investment captures the gain from the market price of risk $\theta$ under a unit noise with square-rooted risk budget $\beta$ as the amplifier. 
Notice that the market price of risk is high when the volatility is low and the budget is high at the same time.
Hence we call the term $\beta ^{1/2}$ as an amplifier to make the investment more efficient.

We also care about the long-term return of $X^\star$ and have the proposition following.
\begin{myProposition}
	Suppose that $X^\star$ has the semi-martingale decomposition $X^\star=F+M$ where $M$ is the continuous martingale part. Then bounded volatility $\sigma_t \in [\underline{\sigma}, \overline{\sigma}]$ together with $||X||_{\mathcal H_{\mathcal S}^2}<\infty$ can ensure that  $M$,  the martingale part of $X^\star$, satisfies 
	\begin{equation*}
	\lim_{t\to \infty} \dfrac{\langle M\rangle _t \log \log t}{t^2}  = \lim_{t\to \infty} \dfrac{(\int_0^t  \hat c ^2 X_s^2 \frac{1}{\sigma_s^2}\myd s) \log \log t }{t^2}= 0
	\end{equation*} 
	and implies the long-term behaviour of $X^\star$
	\begin{equation*}
	\lim_{t\to \infty} \dfrac{M_t}{t}=0 \text{ and } \lim_{t\to \infty} \dfrac{X_t^\star}{t} = \lim _{t\to \infty}\dfrac{\int_0^t \beta_s^{1/2}\theta_s \myd s}{t}.
	\end{equation*}
\end{myProposition}
The limitation above is a straight result of the strong law of large numbers for Brownian motions and martingales. 
Readers can refer to Lemma~1.3.2 in \cite{fernholz_stochastic_2002} for the martingale case.
This proposition states that the long-term average return of $X^\star$ depends on the time-averaged signal $\beta^{1/2}_t\theta_t$.

\subsection{Volatility Budgeting: SABR Model}
\label{Section: SABR}
In this subsection, we demonstrate the relationship between the smile of volatility and the continuous-time risk contribution established Definition~\ref{def:Continuous-time Contribution} by using a risk-budgeted forward interest rate example. 
In light of the smile of the implied volatility, we choose the SABR model to describe the dynamics of this future interest rate.

The SABR model, a stochastic volatility model, is frequently used in the financial industry, particularly in interest rate derivative markets. It was first proposed in \cite{hagan2002managing} to capture the volatility smile for Markovian local volatility model improvement.
Just as the name \emph{'stochastic, alpha, beta, rho'} implies, we consider a single forward interest rate with the stochastic-volatile underlying dynamic
	$$
	\begin{cases}
		\myd F_t = \sigma_t F_t^{\hat \beta} \myd W_{1,t}
		\\
		F_0 = f
	\end{cases}, 
	\begin{cases}
		\myd \sigma_t = \hat \alpha \sigma _t \myd W_{2,t}
		\\
		\sigma_0 = s
	\end{cases}
	$$
where $\myd \langle W_1, W_2\rangle_t = \hat \rho \myd t$. 
Controlling the implied at-the-money volatilities of European options, the parameter $\hat \alpha \geq 0$ acts as the volatility of volatility for the underlying $F$. 
The power coefficient $\hat \beta \in [0,1]$ characterizes the curvature of implied at-the-money volatilities. 
And the remained constant $\hat \rho \in (-1, 1)$ reflects the correlation of two risk sources $W_1$ and $W_2$. 
For convenience, a $2$-dim equivalent Brownian motion $(B_1, B_2)$  is defined by
	\begin{align*}
	\myd W_{1,t} & = \myd B_{1,t},
	\\
	\myd W_{2,t} & = \hat \rho \myd B_{1,t} + \sqrt{1-\hat \rho^2}\myd B_{2,t}.
	\end{align*}

The risk contribution of the specified hedging policy $u$ is now investigated.
The associated terminal variance is given by
	\begin{equation}
	\text{Var}((u\circ F)_T) = \mathbb{ E}[(u\circ F)_T^2]  = \mathbb E [\int_0^T u_t^2 \sigma_t^2 F_t^{2\hat \beta}\myd t]
	\end{equation}
from which we can get the related marginal risk contribution $(c_t)_t$
	\begin{equation}
	c_t = u_t \sigma^2_t F_t^{2\hat \beta} = u_ts^2\mathcal E^2(\hat \alpha \hat \rho B_1)_t\mathcal E^2(\hat \alpha \sqrt{1-\hat \rho ^2}B_2)_tF_t^{2\hat \beta }.
	\end{equation}
Then we can also give the risk contribution $(k_t)_t$
	\begin{equation}
	k_t =u_t c_t = u^2_t s^2\mathcal E^2(\hat \alpha \hat \rho B_1)_t\mathcal E^2(\hat \alpha \sqrt{1-\hat \rho ^2}B_2)_tF_t^{2\hat \beta }
	\end{equation}
where $\mathcal E (\cdot)$ represents the {\Doleans}-Dade exponential martingale of processes.
The parameters in the SABR model are primarily used to calibrate the volatility smile; nevertheless, they can have an impact on our risk contribution in a variety of ways:
\begin{itemize}
    \item The fundamental level of stochastic volatility $s$ has a positive correlation with the contribution of marginal risk;
    \item The stochastic exponential yields a log-normal distribution of volatility changing through $(t,\omega)$ with a scale factor of $\alpha$;
    \item  The correlation $\rho$ reconciles the volatility produced by the rewritten risk resources $B_1$ and $B_2$;
    \item  Naturally, the price of the underlying $F$ has the right to assert that a higher price implies a greater risk contribution;
    \item  Controlling the curvature of implied skewness when confronted with the calibration of options, the parameter $\hat \beta$ also distributes the influence of $F$ across different price levels in an exponential way.
\end{itemize}

Risk-budgeted policies within different information are also of interest, particularly in the case mentioned in Corollary~\ref{coro:Projection}, where we are restricted to some sub-$\sigma$-algebras. 
Suppose that we are dealing with the policy limited to a fixed risk level $\lambda>0$, i.e., 
	$$
	\text{Var}(F_T^u) = \lambda.
	$$
Assume that $\rho = 0$ and a shrinkage filtration $\mathbb H := \{\mathcal H_t\}_t$ generated by $B_1$. The following four cases are projected with less information one-by-one.
\begin{enumerate}
    \item (Purely Risk-Parity) In the view of Theorem~\ref{thm:C-Optimization}, we can indicate the optimal policy associated to risk-parity budget $\dfrac{\lambda}{T}$ for almost all $(t,\omega)$ in $\mathbb F$ by
            $$
            u_t\cdot \big(u_ts^2\mathcal E^2(\hat \alpha  B_2)_tF_t^{2\hat \beta}\big) = \dfrac{\lambda}{T}.
            $$
          Then purely risk-parity policy with total risk $\lambda$ is given by
            $$
            u^\star_t = \sqrt{\dfrac{\lambda}{T}}\dfrac{1}{s\mathcal E(\hat \alpha B_2)_tF_t^{\hat \beta}}.
            $$
    \item  ($\mathbb H$-Projection) Considering that stochastic volatility of underlying $F$ is not easy to perceive, we have a second-best solution restricted on $\mathbb H$ by Corollary~\ref{coro:Projection}
            $$
            u^\prime _t = \sqrt{\dfrac{\lambda}{T}}\dfrac{1}{s\bar F_t^{\hat \beta}}
            $$
          where $\bar F$ is the associated projection of the underlying asset. Comparing this solution to $u^\star$, we cannot distinguish the policy among different $\mathcal E(\hat \alpha B_2)_t(\omega)$ namely the volatility rate $\sigma_t(\omega)$ and $(\sigma_t)_t$ is substituted by its expectation $s$ in the formula above. This projected policy is substantially equivalent to that of the Black-Scholes model, where the volatility rate is assumed constant.
    \item (Non-stochastic Time-varying) Furthermore another shrunk $\sigma$-algebra $\mathcal F_0 \otimes \mathcal B_{[0,T]}$ gives a non-stochastic time-varying policy $u''$ by
          $$
              u''(t)^2 \mathbb E [s^2\mathcal E^2(\hat \alpha  B_2)_tF_t^{2\hat \beta}] = \dfrac{\lambda}{T}
          $$
          for each $t$.
    \item (Naive Single-Period) Finally, a buy-and-hold policy $\bar u$ can give the degenerated single-period solution by the equation
          $$
              \bar u ^2 \mathbb E [\int_0^T s^2\mathcal E^2(\hat \alpha  B_2)_tF_t^{2\hat \beta} \myd t] = \bar u ^2 \text{Var}(F_T) = \lambda.
          $$
\end{enumerate}

These four policies illustrate the distinctions between the single-period case and the continuous-time one. 
Similar to the extended mean value theorem of integrals, single-period policies act like a low-resolution version of continuous-time ones. 
This risk budgeting technique implies that adjusting the position of the underlying, continuous-time risk budgeting can produce more exquisite $(t,\omega)$-pointwise policies associated with more robust targeted risk distributions.

\subsection{Risk Contribution of Continuous-Time Mean-Variance Policy}
\label{Section: Continuous-time MV Example}
 Mean-variance investors are facing a Markowitz's mean-variance minimization target
	$$
	\text {to minimize }J_{MV}(u) = -\mathbb E (X_T^u) + \tau \Var(X_T^u)
	$$
with the coefficient $\tau>0$ representing the risk tolerance to make a balance between the return $\mathbb E (X_T^u)$ and the risk $\Var (X_T^u)$. 
Having found the approach to get the risk contribution of a policy in Section~\ref{Section: Contribution}, we are also interested in how the related risk contribution of continuous-time mean-variance policy behaves.

For simplicity, here we assume the dynamics of two assets are
	\[
	\begin{cases}
	\myd S_t^{0} = r S_t^0\myd t \\
	S_0^0 = s^0,
	\end{cases}
	\begin{cases}
	\myd S_t^1 = S_t^1b\myd t + S_t^1 \sigma \myd W_t
	\\
	S_0^1 = s^1.
	\end{cases}
	\]
Then the wealth process of the self-financing portfolio is specified as
	$$
	\begin{cases}
	\myd X_t = (rX_t + (b-r)M^1_t)\myd t + \sigma M^1_t \myd W_t
	\\
	X_0 = x_0 >0
	\end{cases}
	$$
where $M^1$ is the amount of money allocated on $S^1$. 
According to \cite{zhou2000continuous}, the optimal policy of object $J_{MV} $ is given in a linear feedback manner by
	\begin{align*}
	M^1_t & = M^1(t,X_t)=(\sigma \sigma^\top)^{-1}(b-r)[\dfrac{2\tau x_0 e^{rT} + e^{(b-r)^2(\sigma \sigma ^\top)^{-1}T}}{2\tau}e^{-r(T-t)}-X_t],
	\\
	M^0_t & =M^0(t,X_t) =X_t - M^1_t.
	\end{align*}
For convenience, we denote $M = [M^0, M^1]^\top$ in a vector style where $M^0 = X- M^1$ is the money on the bond $S^0$.
One can also derive the expected return of terminal wealth
	$$
	\mathbb E [X_T] = x_{0} \Big(1 - e^{- \frac{T (b - r)^{2}}{\sigma\sigma^\top}}\Big) e^{\frac{T \left(b - r\right)^{2}}{\sigma\sigma^\top}} + x_{0} e^{T \big(r - \frac{\left(b - r\right)^{2}}{\sigma\sigma^\top}\big)}
	$$
With the help of Corollary~\ref{coro: explicit risk contribution}, we can give the related (marginal) risk contribution of the continuous-time mean-variance policy
	\begin{align*}
	c_t          & = \left[\begin{matrix}2 X r - r \left(\mathbb {E}[X] + x_{0}\right) \left(- P + X\right)\\2 X b + \sigma\sigma^\top \left(P X + Q_t\right)^{2} - b \left(\mathbb {E}[X] + x_{0}\right) \left(P X + Q_t\right)\end{matrix}\right]
	\\
	(u\odot c)_t & =
	\left[\begin{matrix}X r \left(- 2 P + 2 X\right) - r \left(\mathbb {E}[X] + x_{0}\right) \left(- P + X\right)\\X b \left(2 P X + 2 Q_t\right) + \sigma\sigma^\top \left(P X + Q_t\right)^{2} - b \left(\mathbb {E}[X] + x_{0}\right) \left(P X + Q_t\right)\end{matrix}\right]
	\end{align*}
where we denote $P:=-\dfrac{b-r}{\sigma \sigma ^\top}$
and
$Q_t:= \frac{P  e^{- r \left(T - t\right)}}{2  \tau}\left(2 \tau x_{0} e^{T r} + e^{\frac{T \left(b - r\right)^{2}}{\sigma\sigma^\top}}\right)$ for convenience.

Hence the risk contribution $K:= u\odot c$  of two assets has the elements
	\begin{align*}
	K_t^0 =&  \dfrac{2 r \left(\sigma\sigma^\top + b - r\right)}{\sigma\sigma^\top}X_t^2
			+ \dfrac{r}{\sigma \sigma ^\top}\left[ (\sigma \sigma ^\top -b + r)(r^{\rho T}+ e^{(r-\rho)T}) + 2(r-b)e^{rt} \right]x_0 X_t 
		\\
			&+ \dfrac{r \left(- b + r\right) e^{T \rho - T r + r t}}{\sigma\sigma^\top} \dfrac{X_t}{\tau}
				+ \dfrac{br-r^2}{\sigma \sigma ^\top}e^{rt}\left(r^{\rho T}+ e^{(r-\rho)T} \right)x_0^2
		\\  &+ \dfrac{r(b-r)e^{rt}}{2\sigma \sigma ^\top}\left( 1+ e^{2\rho T -rT} \right)\dfrac{x_0}{\tau}
	\end{align*}
and
	\begin{align*}
	K_t^1 = & \dfrac{r^2 - b^2}{\sigma \sigma ^\top} X_t^2
				+  \dfrac{b-r}{\sigma \sigma ^\top}(be^{\rho T} + be^{(r-\rho)T} + 2re^{rt})x_0 X_t
	\\		& +\dfrac{br-r^2}{\sigma \sigma ^\top}e^{\rho T - r(T-t)}\dfrac{X_t}{\tau}
		+ \dfrac{e^{rt}}{\sigma \sigma ^\top} \left((b-r)^2e^{rt} - (b^2 - br)(e^{\rho T}+ e^{(r-\rho)T})  \right)x_0^2 
	\\
	& + \dfrac{e^{-r(T-t)}}{2\sigma \sigma ^\top}\left(2(b-r)^2e^{\rho T + rt} -(b^2 - br)(e^{2\rho T}+ e^{rT})  \right)\dfrac{x_0}{\tau}
	+ \dfrac{1}{\tau ^2}\dfrac{(b-r)^2e^{2(\rho T - rT + rt)}}{4\sigma \sigma ^\top}.
	\end{align*}

In order to have a better understanding on $K$, we illustrate the impact of the exogenous parameters $x_0$ and $ \tau$ in Fig.~\ref{fig:RiskContribution_DifferentParameters}.
	\begin{figure}[H]
	\centering 
	\subfigure[Impact of $x_0$]{
		\includegraphics[width = 2.5 in]{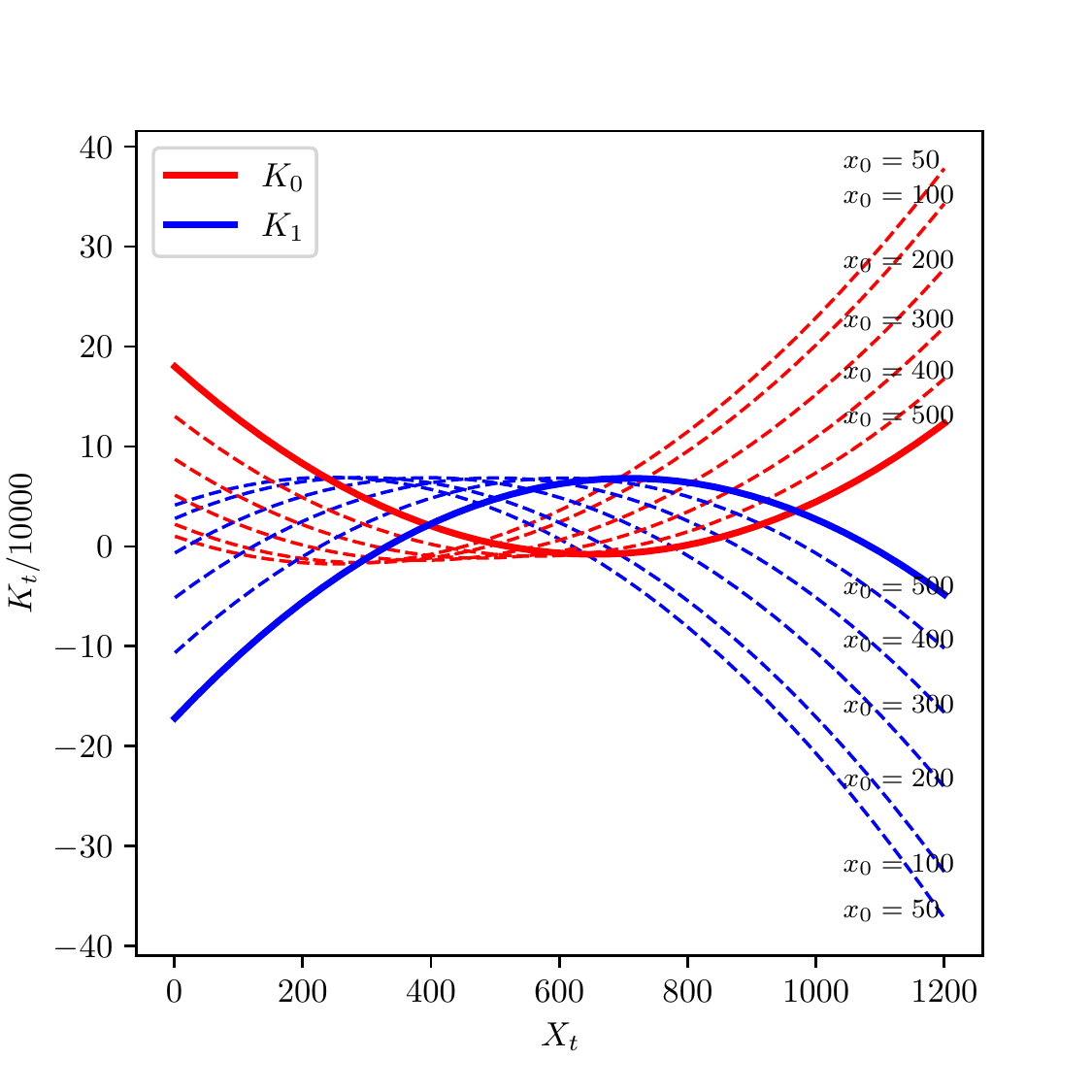}}
	\subfigure[Impact of $\tau$]{
		\includegraphics[width = 2.5 in]{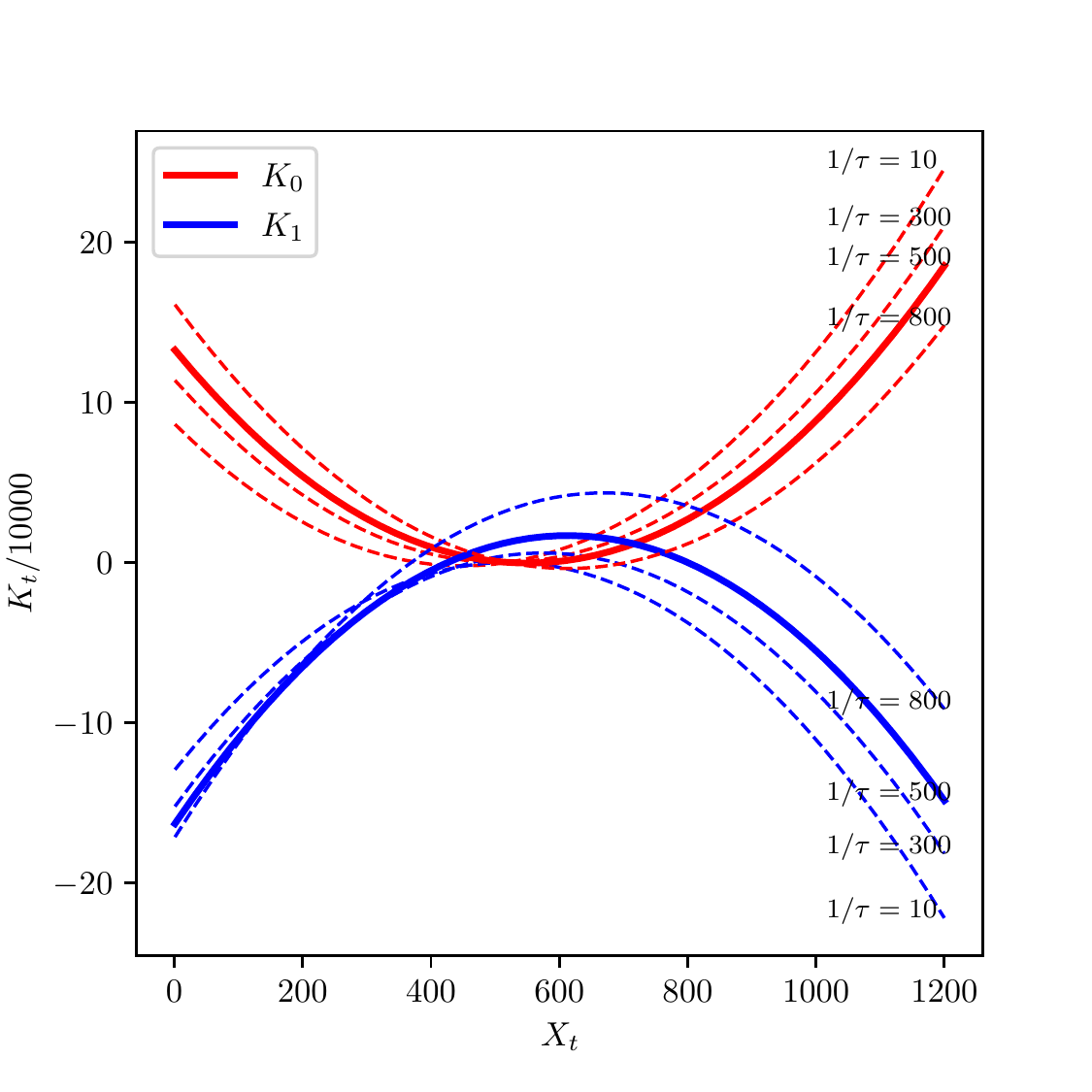}}
	\caption{The Risk Contribution with Different Exogenous Parameters (Baseline Setting: $b = 12\%, \sigma \sigma^\top  = (15\%)^2, r = 6\%, T=1, t = 0.5$)}
	\label{fig:RiskContribution_DifferentParameters}
	\end{figure} 
We can summarize:
 \begin{itemize}
 	\item The risk contributions on each asset are quadratically distributed with respect to the spot price $X_t$. 
     The recession directions of $K^0_t$ and $K^1_t$ are different and specified by the market parameters $r, b, \sigma$ instead of the tolerance $\tau$. 
     As the deviation risk measure suggests, risk contributions of the terminal variance are at a low level around the average of the spot price $\mathbb E X_t$ and a high level far away from the average price.  
    \item In the central area of the spot price $X_t$, $K^1_t$ has a positive risk contribution and vice versa with $K^0_t$. 
     In the extreme circumstances on the spot price $X_t$, however, it shows that the risk contribution $K^1_t$ on the stock is negative, which exceeds our expectation that risky assets are supposed to contribute more in extreme circumstances. 
     This conflict is due to shorting the stock when the spot price $X_t$ is high, or in other words, negative shares on the stock with a positive marginal contribution $c^{1}_t$ lead to the negative cases of $K^1_t$.
    \item When the spot price $X_t$ is varied, the risk tolerance $\tau$ does not affect the second-order state.
     Indeed, when placed in the same position as $\dfrac{1}{x_0}$, $\tau$ can merely contribute to the shift of $K$.
     One can expect that a higher initial wealth parameter $x_0$ brings us a whole lift of risk level without changing the structure of risk contribution, and $\tau$ does the same.
    \item  It appears that policies in the affine feedback form of $X_t$, as well as the mean-variance policy as a specific case, cannot prevent the concentration of risk. 
 \end{itemize}

\section{Concluding Remarks}
\label{Section: Concluding}

We have generalized the concept of risk contribution to the continuous-time case in a differential manner.
The primary purpose of this paper is to give a proper definition of the continuous-time risk contribution for the terminal variance.
The marginal risk contribution is defined from the {\Gateaux} derivative of terminal variance.
The non-degeneration condition ensures the uniqueness of the marginal risk contribution.

The continuous-time risk budget problem is then raised.
Different from the single-period case, the risk budget together with the risk contribution depends not only on the assets but also on $(t,\omega)\in ([0,T],\Omega)$.
The risk-budgeted investment is related to an optimization object, a tradeoff between the minimization of risk and the diversification of shares.
We give an economic interpretation of this optimization in the view of constraints.

We also put three examples to stress the application prospect of the continuous-time risk contribution/budgeting.
\begin{itemize}
    \item A price-regularized risk-parity strategy reproduces the volatility-managed portfolios in \cite{moreira2017volatility}.
    With the help of this risk-budgeted strategy, we give the long-term averaged return of this strategy.
    \item  The risk contribution depends on the concrete dynamics of the assets. 
    We give the risk contribution of SABR dynamics in the second example to describe how the risk contribution is related to the coefficients, which can also be implied by the options.
    \item  As a tool to measure the risk concentration, the risk contribution of the continuous-time mean-variance portfolio is calculated in the last example.
    The result shows that the continuous-time mean-variance policy has a concentration indeed.	
\end{itemize}

\section{Acknowledgements}
\label{Section: Acknowledgment}
The authors would like to thank the referees for their helpful comments.

\bibliographystyle{apalike}
\bibliography{MyBib}
\end{document}